\documentclass[aps,pra,superscriptaddress,twocolumn,footinbib]{revtex4-1}
\usepackage{lmodern}
\usepackage{amsthm,amsmath,amssymb,amsbsy}

\usepackage{color}
\definecolor{darkred}{rgb}{0.8,0.1,0.1}
\definecolor{lightblue}{rgb}{0.1,0.1,0.8}
\usepackage[normalem]{ulem}
\usepackage{graphicx,epic,eepic,epsfig,verbatim,color}
 \usepackage{tikz}
\usepackage[T1]{fontenc}
\usepackage[utf8]{inputenc}
\usepackage[breaklinks]{hyperref}
\usepackage{color, soul}
\usepackage{xcolor}
\usepackage[most,breakable]{tcolorbox}

\definecolor{myhlcolor}{rgb}{1, 1, 0}
\definecolor{myhlcolortwo}{rgb}{1, 1, 0}
\definecolor{myhlcolor}{rgb}{1, 1, 1}
\newtheorem{theorem}{Theorem}

\newtheorem{proposition}[theorem]{Proposition}

\usepackage{cleveref}
\usepackage{graphicx}
\usepackage{xcolor}

\definecolor{darkblue}{RGB}{0,76,156}
\definecolor{darkkblue}{RGB}{0,0,153}
\definecolor{blue2}{RGB}{102,178,255}

\def\endenv{\ifmmode\;\else{\unskip\nobreak\hfil
\penalty50\hskip1em\null\nobreak\hfil\;
\parfillskip=0pt\finalhyphendemerits=0\endgraf}\fi}

\mathchardef\ordinarycolon\mathcode`\:
\mathcode`\:=\string"8000
\def\vcentcolon{\mathrel{\mathop\ordinarycolon}}
\begingroup \catcode`\:=\active
  \lowercase{\endgroup
  \let :\vcentcolon
  }

\makeatletter
\def\resetMathstrut@{%
  \setbox\z@\hbox{%
    \mathchardef\@tempa\mathcode`\[\relax
    \def\@tempb##1"##2##3{\the\textfont"##3\char"}%
    \expandafter\@tempb\meaning\@tempa \relax
  }%
  \ht\Mathstrutbox@\ht\z@ \dp\Mathstrutbox@\dp\z@}
\makeatother
\begingroup
  \catcode`(\active \xdef({\left\string(}
  \catcode`)\active \xdef){\right\string)}
\endgroup
\mathcode`(="8000 \mathcode`)="8000

\newcommand{\nc}{\newcommand}
\nc{\rnc}{\renewcommand}
\nc{\beg}{\begin{equation}}
\nc{\eeq}{\end{equation}}
\nc{\beqa}{\begin{eqnarray}}
\nc{\eeqa}{\end{eqnarray}}
\nc{\lbar}[1]{\overline{#1}}
\nc{\ketbra}[2]{|#1\rangle\!\langle#2|}

\nc{\avg}[1]{\langle#1\rangle}
\nc{\Rank}{\operatorname{Rank}}
\nc{\smfrac}[2]{\mbox{$\frac{#1}{#2}$}}
\nc{\tr}{\operatorname{Tr}}
\nc{\ox}{\otimes}
\nc{\dg}{\dagger}
\nc{\dn}{\downarrow}
\nc{\cA}{{\cal A}}
\nc{\cB}{{\cal B}}
\nc{\cC}{{\cal C}}
\nc{\cD}{{\cal D}}
\nc{\cE}{{\cal E}}
\nc{\cF}{{\cal F}}
\nc{\cG}{{\cal G}}
\nc{\cH}{{\cal H}}
\nc{\cI}{{\cal I}}
\nc{\cJ}{{\cal J}}
\nc{\cK}{{\cal K}}
\nc{\cL}{{\cal L}}
\nc{\cM}{{\cal M}}
\nc{\cN}{{\cal N}}
\nc{\cO}{{\cal O}}
\nc{\cP}{{\cal P}}
\nc{\cQ}{{\cal Q}}
\nc{\cR}{{\cal R}}
\nc{\cS}{{\cal S}}
\nc{\cT}{{\cal T}}
\nc{\cV}{{\cal V}}
\nc{\cU}{{\cal U}}
\nc{\cX}{{\cal X}}
\nc{\cY}{{\cal Y}}
\nc{\cZ}{{\cal Z}}
\nc{\cW}{{\cal W}}
\nc{\csupp}{{\operatorname{csupp}}}
\nc{\qsupp}{{\operatorname{qsupp}}}
\nc{\var}{{\operatorname{var}}}
\nc{\rar}{\rightarrow}
\nc{\lrar}{\longrightarrow}
\nc{\polylog}{{\operatorname{polylog}}}
\nc{\1}{{\mathds{1}}}
\nc{\wt}{{\operatorname{wt}}}
\nc{\av}[1]{{\left\langle {#1} \right\rangle}}
\nc{\supp}{{\operatorname{supp}}}

\def\x{\xi}

\nc{\RR}{{{\mathbb R}}}
\nc{\CC}{{{\mathbb C}}}
\nc{\FF}{{{\mathbb F}}}
\nc{\NN}{{{\mathbb N}}}
\nc{\ZZ}{{{\mathbb Z}}}
\nc{\PP}{{{\mathbb P}}}
\nc{\QQ}{{{\mathbb Q}}}
\nc{\UU}{{{\mathbb U}}}
\nc{\EE}{{{\mathbb E}}}

\nc{\CHSH}{{\operatorname{CHSH}}}

\nc{\be}{\begin{equation}}
\nc{\ee}{{\end{equation}}}
\nc{\bea}{\begin{eqnarray}}
\nc{\eea}{\end{eqnarray}}
\nc{\Hom}[2]{\mbox{Hom}(\CC^{#1},\CC^{#2})}
\nc{\rU}{\mbox{U}}

\nc{\ob}[1]{#1}

\nc{\SEP}{{\text{SEP}}}
\nc{\NS}{{\text{NS}}}
\nc{\LOCC}{{\text{LOCC}}}
\nc{\PPT}{{\text{PPT}}}
\nc{\EXT}{{\text{EXT}}}
\nc{\Sym}{{\operatorname{Sym}}}

\nc{\ERLO}{{E_{\text{r,LO}}}}
\nc{\ERLOCC}{{E_{\text{r,LOCC}}}}
\nc{\ERPPT}{{E_{\text{r,PPT}}}}
\nc{\ERLOCCinfty}{{E^{\infty}_{\text{r,LOCC}}}}
\nc{\Aram}{{\operatorname{\sf A}}}

\usepackage{tikz}
\usetikzlibrary{arrows}

\makeatletter
\def\grd@save@target#1{%
  \def\grd@target{#1}}
\def\grd@save@start#1{%
  \def\grd@start{#1}}
\tikzset{
  grid with coordinates/.style={
    to path={%
      \pgfextra{%
        \edef\grd@@target{(\tikztotarget)}%
        \tikz@scan@one@point\grd@save@target\grd@@target\relax
        \edef\grd@@start{(\tikztostart)}%
        \tikz@scan@one@point\grd@save@start\grd@@start\relax
        \draw[minor help lines,magenta] (\tikztostart) grid (\tikztotarget);
        \draw[major help lines] (\tikztostart) grid (\tikztotarget);
        \grd@start
        \pgfmathsetmacro{\grd@xa}{\the\pgf@x/1cm}
        \pgfmathsetmacro{\grd@ya}{\the\pgf@y/1cm}
        \grd@target
        \pgfmathsetmacro{\grd@xb}{\the\pgf@x/1cm}
        \pgfmathsetmacro{\grd@yb}{\the\pgf@y/1cm}
        \pgfmathsetmacro{\grd@xc}{\grd@xa + \pgfkeysvalueof{/tikz/grid with coordinates/major step}}
        \pgfmathsetmacro{\grd@yc}{\grd@ya + \pgfkeysvalueof{/tikz/grid with coordinates/major step}}
        \foreach \x in {\grd@xa,\grd@xc,...,\grd@xb}
        \node[anchor=north] at (\x,\grd@ya) {\pgfmathprintnumber{\x}};
        \foreach \y in {\grd@ya,\grd@yc,...,\grd@yb}
        \node[anchor=east] at (\grd@xa,\y) {\pgfmathprintnumber{\y}};
      }
    }
  },
  minor help lines/.style={
    help lines,
    step=\pgfkeysvalueof{/tikz/grid with coordinates/minor step}
  },
  major help lines/.style={
    help lines,
    line width=\pgfkeysvalueof{/tikz/grid with coordinates/major line width},
    step=\pgfkeysvalueof{/tikz/grid with coordinates/major step}
  },
  grid with coordinates/.cd,
  minor step/.initial=.2,
  major step/.initial=1,
  major line width/.initial=2pt,
}
\makeatother

\tikzset{
  treenode/.style = {align=center, inner sep=0pt, text centered,
    font=\sffamily},
  arn_n/.style = {treenode, circle, white, font=\sffamily\bfseries, draw=black,
    fill=black, text width=1.5em},
  arn_r/.style = {treenode, circle, red, draw=red, 
    text width=1.5em, very thick},
  arn_x/.style = {treenode, rectangle, draw=black,
    minimum width=0.5em, minimum height=0.5em}
}

\usepackage{pifont}
\usepackage{graphicx}
\usepackage{bm,bbm}%
\usepackage{braket}
\usepackage{enumerate}
\usepackage{subcaption,caption}
\usepackage{booktabs,multirow}
\usepackage{mathtools,bbding}
\usepackage[percent]{overpic}
\usepackage{microtype}

\usepackage{newpxtext,newpxmath}

\usepackage[most,breakable]{tcolorbox}

\AtBeginDocument{
\heavyrulewidth=.08em
\lightrulewidth=.05em
\cmidrulewidth=.03em
\belowrulesep=.65ex
\belowbottomsep=0pt
\aboverulesep=.4ex
\abovetopsep=0pt
\cmidrulesep=\doublerulesep
\cmidrulekern=.5em
\defaultaddspace=.5em
}

\usepackage{etoolbox}
\hypersetup{linktoc=all}

\nc{\MIO}{{\text{\rm MIO}}}
\nc{\DIO}{{\text{\rm DIO}}}
\nc{\SIO}{{\text{\rm SIO}}}
\nc{\IO}{{\text{\rm IO}}}

\let\oldproofname\proofname
\renewcommand{\proofname}{\rm\bf{\oldproofname}}

\makeatletter
\renewenvironment{proof}[1][\proofname]{%
  \vspace{-\topsep}
  \pushQED{\qed}
  \normalfont
  \topsep6\p@\@plus6\p@\relax
  \trivlist\item[\hskip\labelsep\bfseries#1\@addpunct{.}]\ignorespaces}{\popQED\endtrivlist\@endpefalse}
\makeatother

\begin{document}
\nocite{*} 
\title{Weight-Based Measure of Quantum Memory as a Universal and Operational Benchmark}
\author{Jinghang Zhang}
\affiliation{School of Artificial Intelligence and Computer Science, Shaanxi Normal University, Xi'an, 710062, China}
\author{Yu Luo}
\email{penroseluoyu@gmail.com}
\affiliation{School of Artificial Intelligence and Computer Science, Shaanxi Normal University, Xi'an, 710062, China}
\begin{abstract}
Quantum memory plays a critical role in quantum communication, sensing, and computation. However, studies on quantum memory under a unified benchmarking framework remain scarce. In this paper, we propose a weight-based quantifier as a benchmarking method to evaluate the performance advantage of quantum memory in nonlocal exclusion tasks. We establish a general lower bound for the weight-based measure of quantum memory. Moreover, this measure provides fundamental theoretical bounds for transforming a general channel into an ideal quantum memory. Finally, we present explicit calculations of the weight-based quantifier for various channels, including unitary channels, depolarizing channels, maximal replacement channels, stochastic damping channels, and erasure channels.
\end{abstract}
\date{\today}
\maketitle
\section{Introduction} 
Memory is a fundamental component of information processing, playing a vital role in both classical and quantum computing~\cite{nicolas2014quantum,freer2017single,wang2019efficient,hosseini2011high,hsiao2018highly,vernaz2018highly,lvovsky2009optical,goronkin2004high,vieira2024entanglement,simnacher2019certifying,yuan2021universal,chang2024visually}. The primary function of a quantum memory is to preserve an input quantum state with minimal loss over a designated period of time. However, due to the existence of quantum superposition, quantum memory is fundamentally distinct from its classical counterpart. For instance, while classical memory can store the outcomes of measurements on orthogonal basis states such as $\{\ket{0}, \ket{1}\}$, the act of measurement in quantum systems collapses the state, thereby eliminating any information about the original quantum superposition. As a result, it becomes impossible to distinguish between superposition states such as $(\ket{0} + \ket{1})/\sqrt{2}$ and $(\ket{0} - \ket{1})/\sqrt{2}$, which yield identical statistical outcomes under measurements in the computational basis~\cite{nielsen2010quantum,wootters1982single}. 
This fundamental distinction highlights the necessity of developing operationally meaningful methods to evaluate and compare the performance of quantum memories across diverse physical implementations and use cases. However, unified benchmarking frameworks and quantification methods for evaluating the performance of quantum memories remain relatively scarce~\cite{heshami2016quantum,lvovsky2009optical,yuan2021universal,chang2024visually}. 

Recently, the authors of Ref.~\cite{yuan2021universal,ku2022quantifying} proposed treating quantum memory as a physical resource and investigated its robustness. They demonstrated that this resource-based approach exhibits broad applicability across multiple scenarios, including memory synthesis, classical simulation overhead, and channel discrimination
tasks, making it a promising candidate for a universal benchmarking framework.

A natural question that arises is whether there exist other benchmarks for characterizing the memory of a channel.
In this work, we focus on a specific quantification method for quantum memory — the weight-based measure~\cite{uola2020all,ducuara2020operational,bu2018asymmetry,streltsov2010linking}. The weight-based measure evaluates the quantum resources of a target quantum memory by calculating the minimum amount of noise required to degrade it into a "free memory" incapable of quantum storage. By definition, an identity channel serves as the benchmark for perfect quantum memory as it perfectly preserves all quantum information. Conversely, we consider \textit{entanglement-breaking} (\textrm{EB}) channels to be entirely incapable of storing quantum information, which is regarded as "free memory"~\cite{tabia2024super,vieira2024entanglement,buscemi2011entanglement,horodecki2003entanglement}. We find that this measure demonstrates significant operational advantages in non-local exclusion tasks. Specifically, the weight of a quantum memory is directly linked to its ability to surpass classical strategies in such tasks~\cite{uola2020all}. This is crucial for the success of specific quantum information processing protocols, such as tasks requiring precise coordination or information authentication in secure communication or distributed computing scenarios~\cite{brassard2005quantum,gottesman2001quantum,bennett2014quantum,pirandola2020advances,10236453,11131294,jing2025circuit,piveteau2025circuit,piveteau2025simulating}. This advantage in mutually exclusive tasks forms one of the core focuses of our research, as it directly links an abstract measure to practical quantum information processing capabilities.

This study further reveals a series of important properties of the weight-based measure, including its monotonicity and convexity. We demonstrate that this measure is associated with quantum state optimization problems, and this optimization can be achieved by analyzing the channel's Choi state. Furthermore, by introducing non-local exclusion tasks, we showcase the operational meaning of this measure. We also examine the unavoidable lower bounds of error arising from free super-channel constraints when distilling arbitrary noisy channels into ideal unitary or replacement channels, thereby revealing a close connection between the weight-based measure of quantum memory and the performance of these purification tasks. Besides, we establish a general lower bound for the weight-based measure of quantum memory. Finally, by calculating the weights for some representative quantum memory channels (including unitary channels, depolarizing channels, maximal replacement channels, stochastic damping
channels, and erasure channels), we find that the weight-based measure can serve as an effective, albeit relatively coarse-grained, performance evaluation metric for quantum memory, providing valuable references for future quantum memory design and optimization.

This paper is organized as follows. In Section \ref{sec:Preliminary}, we first introduce the necessary notation and settings we need. In Section \ref{sec:mainresults}, we will introduce the weight-based measure of quantum memory and introduce the  quantum games to test the power of a memory and evaluate the performance advantage of quantum memory in nonlocal exclusion tasks. Besides, we establish a general lower bound for the weight-based measure of quantum
memory. In Section \ref{sec:fundamental}, we introduce fundamental theoretical bounds for transforming a general channel into an ideal quantum
memory. In Section \ref{sec:examples}, we compute the weight-based measure for several example quantum channels. Then, we summarize our results in Section \ref{sec:conclusion}.
\section{Preliminaries} \label{sec:Preliminary}

\subsection{Notations}
Throughout this paper, we adopt most of the notations used in Refs.~\cite{gour2019comparison,luo2025one,chitambar2019quantum,luo2022coherence}. All Hilbert spaces $\mathscr{H}$ considered are finite dimensional. We will denote all the dynamical systems and their corresponding Hilbert spaces by $A,B$, etc., and all the static systems and their corresponding Hilbert spaces by $A_0,A_1,B_0,B_1$, etc. 
The set of all bounded operators acting on system $A_1$ is denoted by $\mathrm{B}(A_1)$, and the set of all density operators on $A_1$ is denoted by $\mathrm{D}(A_1)$. Throughout this work, density operators will be represented by lowercase Greek letters such as $\rho$, $\sigma$, and $\tau$. We use the notation $A_1B_1$ to represent a composite system $A_1\otimes B_1$, which denote the dimension of a system $A_1$ as $|A_1|$ and the dimension of a system $B_1$ as $|B_1|$. The notion $\tilde{A}$ denotes a system with the same dimension of $A$, the dimension of $A$ is generally denoted by $d$ unless otherwise specified. A pair of systems such that $A := (A_0 \to A_1)$ where $A_0$ and $A_1$ represent the input and output systems, respectively. Moreover, the set of all linear maps from $\textrm B(A_0)$ to $\textrm B(A_1)$ will be denoted as $\textrm{L}(A):=\textrm{L}(A_0\to A_1)$, among which all completely-positive and trace-preserving maps ($\operatorname{CPTP}$) are denoted as $\operatorname{CPTP}(A_0\to A_1)$. A $\operatorname{CPTP}$ map is also called a quantum channel. The identity channel in $\operatorname{CPTP}(A_0\to A_1)$ will be denoted by $\textrm{id}_{A_0}$. We will use calligraphic letters (\( \mathcal{E}, \mathcal{M}, \mathcal{N}, \mathcal{U}, \mathcal{V} \), etc.) to denote quantum channels. 
For any channel $\mathcal{E}_{A_0\to A_1}\in \operatorname{CPTP}(A_0\to A_1)$, the corresponding Choi state is defined as $\mathcal{J}^{\mathcal{E}_{A_0 \to A_1}}_{A_0 A_1} = \mathrm{id}_{A_0} \otimes \mathcal{E}_{A_0\to A_1}(\Psi^+_{A_0 \tilde{A}_0})$, where $\Psi^+_{A_0 \tilde{A}_0} = \frac{1}{d}\sum_{i,j=1}^{d} |ii\rangle\langle jj|_{A_0 \tilde{A}_0}$ is the normalized maximally entangled states for a bipartite system $A_0\tilde{A_0}$, and $d$ is the dimension of system $A_0$. The action of a quantum channel will usually be denoted by parentheses, as in \( \mathcal{N}(\rho) \). We use $\langle A, B\rangle=\operatorname{Tr}(A^{\dagger} B)$ for the Hilbert-Schmidt inner product between operators. 


We will use capital Greek letters like $\Omega,\Theta,\Lambda$ to denote super-operations (supermaps). The action of a super-operations will be represented with square brackets, as in \( \Omega[\mathcal{N}] \) and \( \Omega[\mathcal{E}] \). Superchannels are special super-operations that transform a quantum channel into another quantum channel. A superchannel $\Lambda$ acts on a channel $\mathcal{N}^{A_0 \to A_1}$ to simulate anther channel $\mathcal{M}^{B_0 \to B_1}$. This process can be described as:
\begin{equation}
\Lambda\left[\mathcal{N}^{A_0 \to A_1}\right]=\mathcal{V}^{A_1 E \to B_1} \circ\left(\mathcal{N}^{A_0 \to A_1} \otimes \operatorname{id}^{E}\right) \circ \mathcal{U}^{B_0 \to E A_0},
\end{equation}
where $\mathcal{U}^{B_0 \to E A_0 }$ and $\mathcal{V}^{A_1 E \to B_1}$ are pre-processing and post-processing channels, $\operatorname{id}^E$ is the identity channel on an ancillary system $E$.

As demonstrated in Ref.~\cite{gour2019comparison}, a supermap $\Theta$ is a superchannel if and only if the $\operatorname{CPTP}$ map $\Gamma^{ B_0A_1 \rightarrow 
A_0B_1}$ which corresponds to the Choi matrix
\begin{equation}
\mathcal{J}_{\Theta}=\left(\operatorname{id}^{B_0A_1} \otimes \Gamma_{\Theta}^{B_0A_1 \rightarrow A_0B_1} \right)\left(\Phi_{+}^{B_0 B_0'} \otimes \Phi_{+}^{A_1 A_1'}\right),
\end{equation}
can be expressed as
\begin{equation}
\Gamma_{\Theta}^{B_0A_1 \rightarrow 
A_0B_1} = \left( \textrm{id}^{A_0} \otimes \mathcal{V}^{A_1 E \to B_1} \right) \circ ( \text{id}^{A_1} \otimes \mathcal{U}^{B_0 \to E A_0 }). 
\end{equation}
Here, $\Gamma_{\Theta}^{B_0A_1 \rightarrow 
A_0B_1}$ is the equivalent quantum channel that takes inputs from systems $B_0$ and $A_1$ and produces outputs on systems $A_0$ and $B_1$. Its structure is a composition of the pre-processing channel $\mathcal{U}$ acting on the $B_0$ part and the post-processing channel $\mathcal{V}$ acting on the $A_1$ part (along with the ancilla E).
Its illustration is given in Fig.~\ref{fig:superchannel}.

As shown in Ref.~\cite{saxena2020dynamical}, any superchannel $\Theta^{A\to B}$ has a Kraus decomposition, i.e., an operator sum representation, 
\begin{equation}
    \Theta_{A\to B} = \sum_{x=1}^n\Theta^x_{A\to B},
\end{equation} 
where the Choi matrix of each $\Theta^x$ has rank one, we use this property to prove monotonicity below.

We will omit the subscript indicating the system when the system on which a channel or superchannel acts is explicitly specified. For instance, we may write \(\mathcal{N}\) in place of \(\mathcal{N}^A \in \mathrm{CPTP}(A_0 \to A_1)\) or \(\mathcal{N}^{A_0 \to A_1} \in \mathrm{CPTP}(A_0 \to A_1)\), and \(\Theta\) instead of \(\Theta_{A \to B}\).


\begin{figure}[htbp]
    \centering
    \includegraphics[width=0.5\textwidth]{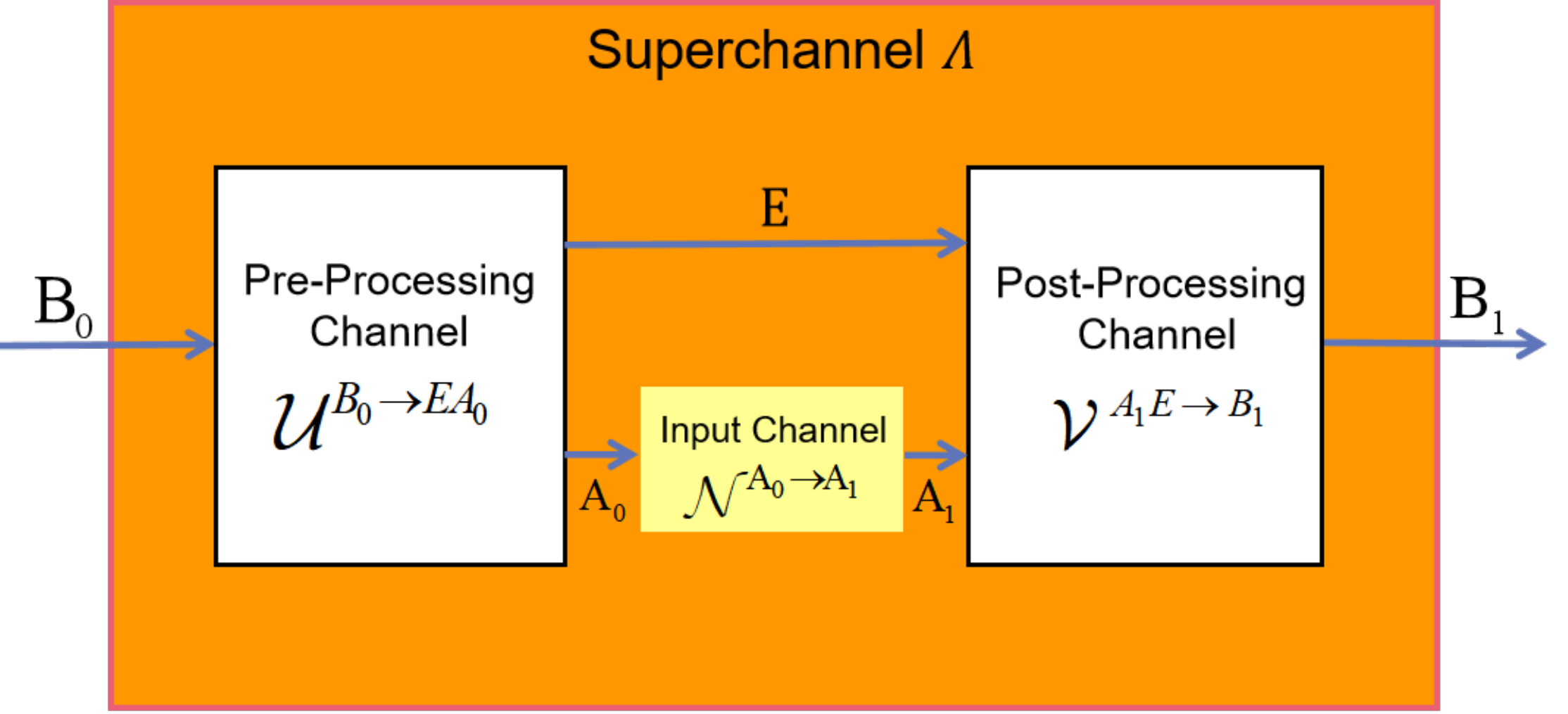} 
    \caption{Achieving a superchannel and its effect on the input channel through pre-processing and post-processing of the channel.}
    \label{fig:superchannel}
\end{figure}

\subsection{ Resource theories of quantum memory}\label{subsec3}
In this section, we will introduce the resource theory of quantum memories~\cite{yuan2021universal}. A resource theory of quantum memory is denoted as a 2-tuple $\mathfrak{R} = (\mathfrak{F}, \mathfrak{O})$, where $\mathfrak{F}$ is the set of free memories and $\mathfrak{O}$ is the set of free super-operations. The set of free memories $\mathfrak{F}$ comprises all channels incapable of quantum storage, namely the EB channels. The set of free super-operations $\mathfrak{O}$ consists of all super-operations that map free channels to free channels~\cite{liu2020operational,luo2022coherence,luo2024epsilon,luo2025one,gour2019comparison,liu2019resource}. Further details regarding $\mathfrak{F}$ and $\mathfrak{O}$ will be provided below.

Any functional $C: \operatorname{CPTP}\to \mathbb{R}^+$ is considered a valid measure for the resource theory of quantum memory, if it satisfies the following two basic requirements:
\begin{enumerate}
    \item[\textbf{[M1]}] (Non-negativity): $C(\mathcal{N}) \geq 0$, with equality holding if $\mathcal{N} \in \mathfrak{F}$;
    \item[\textbf{[M2]}] (Monotonicity): $C(\mathcal{N}) \geq C(\Theta[\mathcal{N}])$ for all free superchannels $\Theta \in \mathfrak{O}$.
\end{enumerate}


 Free memories are represented by $\textrm{EB}$ channels, which are also known as measure-and-prepare channels, as illustrated in Fig.~\ref{fig:EB_channel}. An $\textrm{EB}$ channel $\mathcal{N}^{A_0\to A_1}$ transforms a quantum state $\rho^{A_0}$ from system A to system B as follows:
\begin{equation}\label{eq:EB_channels}
\mathcal{N}^{A_0\to A_1}(\rho^{A_0})=\sum_i\operatorname{Tr}_{A_0}[\rho^{A_0}\operatorname{M}^{A_0}_i]\sigma^{A_1}_i,
\end{equation}
Here, $\{\operatorname{M}^{A_0}_i\}$ is a POVM satisfying $\operatorname{M}^{A_0}_i\ge 0$ and $\sum_i\operatorname{M}^{A_0}_i=I^{A_0}$.
EB channels are considered to be free resource because only classical information is stored and forwarded from system $A_0$ to system $A_1$ and they are mathematically synonymous with classical memories~\cite{yuan2021universal}.

Free super-operations are higher-order transformations that act on quantum channels. Within this framework, free super-operations are those that only transmit classical information~\cite{yuan2021universal}, 
\begin{equation}
    \Omega[\mathcal{N}^{A_0 \to A_1}] = \mathcal{V}^{A_1E \to B_1} \circ \mathcal{N}^{A_0 \to A_1} \circ \mathcal{E}_{deph}^E \circ \mathcal{U}^{B_0 \to EA_0}. 
\end{equation}
Here $\mathcal{U}^{B_0 \to EA_0}$ and $\mathcal{V}^{A_1E \to B_1}$ are arbitrary quantum channels and $\mathcal{E}_{deph}^E$ is a dephasing channel that enforces system $E$ to be classical.
For mathematical simplicity, free super-operations can be defined as resource non-generating super-operations.
This means that a super-operation $\Omega$ is considered free if and only if it maps any free channel back into the set of free channels~\cite{yuan2021universal}:
\begin{equation}
\mathfrak{O}=\{\Omega: \Omega[\mathcal{N}] \in \mathfrak{F}, \forall \mathcal{N} \in \mathfrak{F}\} .
\end{equation}
Furthermore, this restriction ensures that any valid resource measure must be non‐increasing under all $\Omega\in\mathfrak{O}$, thereby providing a consistent notion of monotonicity. By limiting our transformations to $\mathfrak{O}$, we guarantee that no entanglement resource can be generated in the course of processing.

\begin{figure}[htbp]
    \centering
    \includegraphics[width=0.4\textwidth]{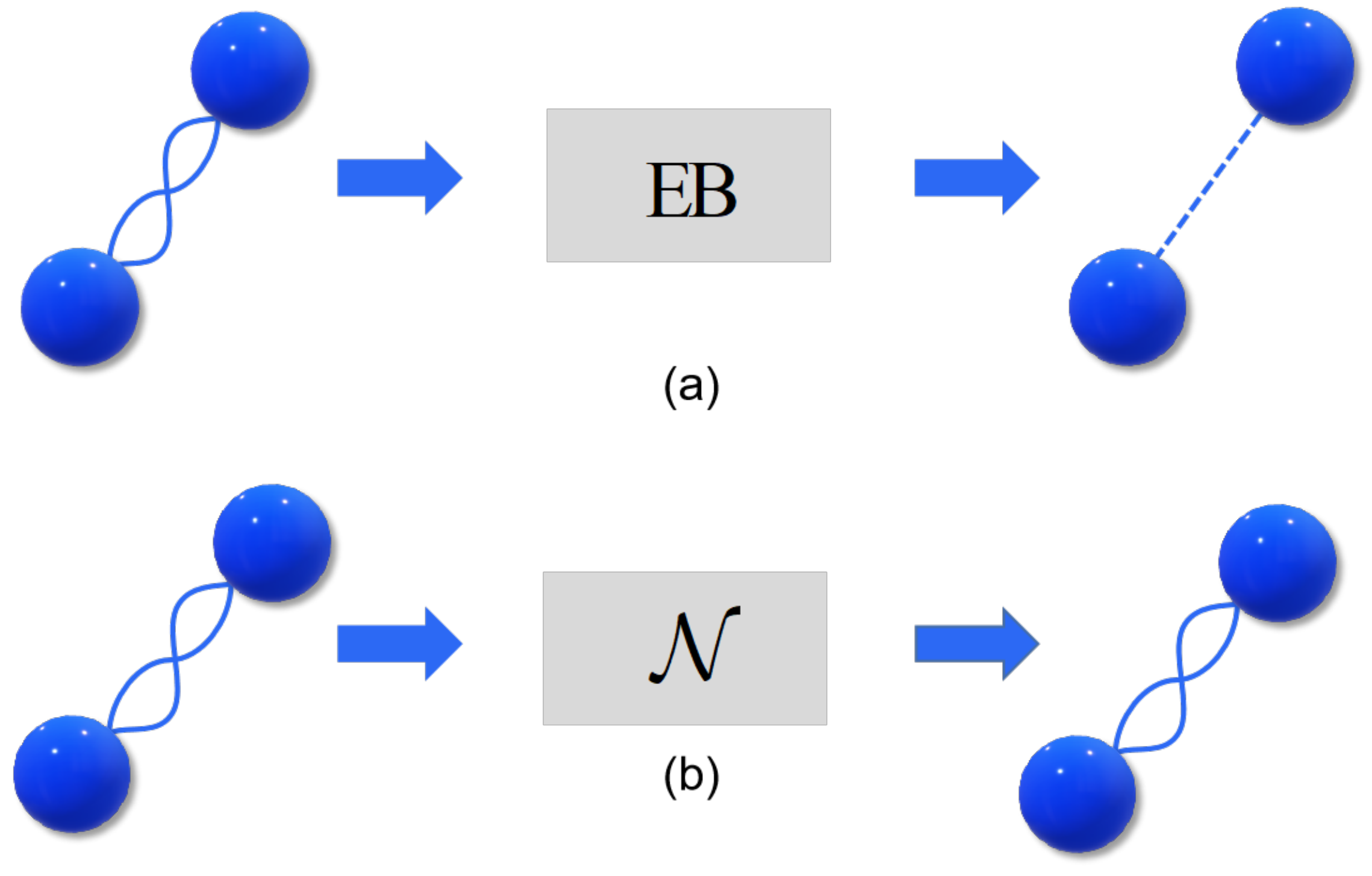} 
    \caption{Different quantum memory resource theories. (a) shows an EB memory, whose function is to eliminate all quantum entanglement, while (b) depicts an ideal memory that can losslessly preserve the original quantum state.}
    \label{fig:EB_channel}
\end{figure}

\section{Main Results} \label{sec:mainresults}

\subsection{Weight-based measure of quantum memory}
For a channel $\mathcal{N}$, the weight-based measure of quantum memory can be defined as
\begin{equation}\label{eq:eq1}
C_w(\mathcal{N}) = \min\{s \ge 0 : \mathcal{N} \ge (1-s)\mathcal{M}, \mathcal{M} \in \textrm{EB}\}.
\end{equation}

We can use the similar method mentioned in Ref.~\cite{wilde2020amortized,luo2022coherence}, the weight-based measure of channel in Eq.~\eqref{eq:eq1} can also be written via Choi state:
\begin{equation}\label{eq:eq2}
C_w(\mathcal{N}) = \min\{s \ge 0 : \mathcal{J}_{\mathcal{N}} \ge (1-s)\mathcal{J}_{\mathcal{M}}, \mathcal{M} \in \textrm{EB}\}.
\end{equation}

This can be equivalently recast as

\begin{equation}\label{eq:choi_state}
\begin{aligned}
C_w(\mathcal{N})= & \min s  \\
\text { s.t. } & \mathcal{J}_\mathcal{N} \ge (1-s)\mathcal{J}_{\mathcal{M}} ,\\
& \mathcal{M} \in \operatorname{EB},
\end{aligned}
\end{equation}
where we use $\mathcal{J}_\mathcal{M}=\frac{1}{d}\sum\ketbra{i}{i}\otimes\mathcal{M}(\ketbra{i}{i})$ denotes the Choi state of EB channel, which are separable states.

\begin{proposition}\label{pro:pro2}
The weight-based measure of quantum memory satisfies monotonicity, i.e.,
\begin{equation}
C_w(\mathcal{N}) \ge \sum_x{p_xC_w(\Theta_x[\mathcal{N}]/p_x)},
\end{equation}
where $\Theta=\sum_{x}\Theta^x_{A\to B}$ is a superchannel and $p_x=\frac{1}{|B_0|}Tr[\mathcal{J}_{\Theta_x[\mathcal{N}]}]$.
\end{proposition}
\begin{proof}
By definition of the coherence weight, there exists an optimal decomposition of the Choi matrix of $\mathcal{N}$
\begin{equation}
J_{\mathcal{N}}=(1-\tilde{s}) J_{\mathcal{E}}+\tilde{s} J_{\mathcal{M}},
\end{equation}
where $\mathcal{E}\in \operatorname{EB}$ and $\mathcal{M},\mathcal{N} \in \operatorname{CPTP}$.
For each sub-superchannel $\Theta_x$, its action on the Choi matrix can be viewed as a linear, positive map, which we denote by $\Pi_{\Theta_{x}}$. Hence, 
\begin{equation}
\begin{aligned}
\mathcal{J}_{\Theta[N]} & =\Pi_{\Theta}\left(\mathcal{J}_{N}\right), \\
& =\Pi_{\Theta}\left((1-\tilde{s}) \mathcal{J}_{\mathcal{E}}+\tilde{s} \mathcal{J}_{M}\right), \\
& =(1-\tilde{s}) \Pi_{\Theta}\left(\mathcal{J}_{\mathcal{E}}\right)+\tilde{s} \Pi_{\Theta}\left(\mathcal{J}_{M}\right), \\
& =(1-\tilde{s}) \mathcal{J}_{\Theta[\mathcal{E}]}+\tilde{s} \mathcal{J}_{\Theta[M]},\\
& = \sum_x((1-\tilde{s}) \mathcal{J}_{\Theta_x[\mathcal{E}]}+\tilde{s} \mathcal{J}_{\Theta_x[M]}).
\end{aligned}
\end{equation}
According to the method described in Ref.~\cite{luo2022coherence}, we suppose
\begin{equation}
\mathcal{J}_{\mathcal{E}}^{x}=\frac{1}{\left(1-s_{x}\right) p_{x}}(1-\tilde{s}) \mathcal{J}_{\Theta_{x}\left[{\mathcal{E}}\right]},
\end{equation}
\begin{equation}
\mathcal{J}_{\mathcal{M}}^{x}=\frac{1}{s_{x} p_{x}} \tilde{s} \mathcal{J}_{\Theta_{x}\left[\tilde{\mathcal{M}}\right]},
\end{equation}
\begin{equation}
s_{x}=\frac{1}{p_{x}} \frac{\tilde{s}}{\left|B_{0}\right|} \operatorname{Tr}\left[\mathcal{J}_{\Theta_{x}\left[\tilde{\mathcal{M}}\right]}\right].
\end{equation}
Then, we have
\begin{equation}
\frac{\mathcal{J}_{\Theta_{x}[N]}}{p_x}
= (1 - s_x)\mathcal{J}_{\mathcal{E}}^{x}
+ s_x\mathcal{J}_{\mathcal{M}}^{x}.
\end{equation}
By definition of the weight-based measure, we can conclude $C_w(\Theta_{x}[\mathcal{N}]/p_x)\le s_x$,which implies
\begin{equation}
\sum_xp_xC_w(\Theta_{x}[\mathcal{N}]/p_x)\le\sum_xp_xs_x=C_w(\mathcal{N}).
\end{equation}
\end{proof}

\begin{proposition}\label{pro:pro1}
The weight-based measure of quantum memory is convex, i.e.,
\begin{equation}
C_w(p\mathcal{N}_1+(1-p)\mathcal{N}_2) \le pC_w(\mathcal{N}_1)+(1-p)C_w(\mathcal{N}_2),
\end{equation}
where $0\le p\le1$ and $\mathcal{N}_1,\mathcal{N}_2\in \operatorname{CPTP}$.
\end{proposition}
\begin{proof}
We need to show that the weight-based measure $C_w$ satisfies the convexity property. Consider two quantum channels $\mathcal{N}_1$ and $\mathcal{N}_2$, and define a new channel $\mathcal{N}_p=p\mathcal{N}_1+(1-p)\mathcal{N}_2$. The goal is to prove that 
\begin{equation}
C_w(\mathcal{N}_p) \le pC_w(\mathcal{N}_1)+(1-p)C_w(\mathcal{N}_2).
\end{equation}
Similarly, the Choi matrix $\mathcal{J}_{\mathcal{N}_p}=p\mathcal{J}_{\mathcal{N}_1}+(1-p)\mathcal{J}_{\mathcal{N}_2}$.

Therefore, for any $s_1$ and $s_2$ corresponding to $\mathcal{N}_1$ and $\mathcal{N}_2$, we can write
\begin{equation}
\mathcal{J}_{\mathcal{N}_{1}} \geq\left(1-s_{1}\right) \mathcal{J}_{\mathcal{M}_{1}}, \quad \mathcal{J}_{\mathcal{N}_{2}} \geq\left(1-s_{2}\right) \mathcal{J}_{\mathcal{M}_{2}},
\end{equation}
where $\mathcal{M}\in \operatorname{EB}$ is a sepatable Choi state. Then we have
\begin{equation}
\begin{aligned}
\mathcal{J}_{\mathcal{N}_{p}}&=p \mathcal{J}_{\mathcal{N}_{1}}+(1-p) \mathcal{J}_{\mathcal{N}_{2}}, \\
&\geq p\left(1-s_{1}\right) \mathcal{J}_{\mathcal{M}_{1}}+(1-p)\left(1-s_{2}\right) \mathcal{J}_{\mathcal{M}_2},\\ 
&= (1-\left(p s_{1}+(1-p) s_{2}\right))\mathcal{J_M}.
\end{aligned}
\end{equation}
Since $p \mathcal{J}_{\mathcal{M}_{1}}+(1-p) \mathcal{J}_{\mathcal{M}_2}$ is still a divisible Choi state, then it can be obtained by definition
\begin{equation}
\begin{aligned}
C_{w}\left(\mathcal{N}_{p}\right) &\leq p s_{1}+(1-p) s_{2},\\
&=p C_{w}\left(\mathcal{N}_{1}\right)+(1-p) C_{w}\left(\mathcal{N}_{2}\right).
\end{aligned}
\end{equation}
\end{proof}

\begin{proposition}\label{pro:pro3}
For any two channels $\mathcal{N}_1\in \operatorname{CPTP}(A)$ and $\mathcal{N}_2\in \operatorname{CPTP}(\Tilde{A})$, the weight-based measure of quantum memory satisfies
\begin{equation}\label{eq:eq3}
C_w(\mathcal{N}_1\otimes\mathcal{N}_2) \le C_w(\mathcal{N}_1)+C_w(\mathcal{N}_2)-C_w(\mathcal{N}_1)C_w(\mathcal{N}_2),   
\end{equation}
\begin{equation}\label{eq:otimes}
C_w(\mathcal{N}_1\circ\mathcal{N}_2) \le C_w(\mathcal{N}_1)+C_w(\mathcal{N}_2)-C_w(\mathcal{N}_1)C_w(\mathcal{N}_2).
\end{equation}
\end{proposition}
\begin{proof}
First, we suppose
\begin{equation}
s_i = C_w(\mathcal{N}_i),
\end{equation}
and choose entanglement-breaking channels $\mathcal{E}_i$ such that 
\begin{equation}
\mathcal{N}_i\ge(1-s_i)\mathcal{E}_i.
\end{equation}
Then we can write 
\begin{align}
\mathcal{N}_{1} \otimes \mathcal{N}_{2} &\geq\left(1-s_{1}\right) \mathcal{E}_{1} \otimes\left(1-s_{2}\right) \mathcal{E}_{2},\\
&=\left(1-s_{1}\right)\left(1-s_{2}\right)\left(\mathcal{E}_{1} \otimes \mathcal{E}_{2}\right),
\end{align}
where $\mathcal{E}_1\otimes\mathcal{E}_2$ remains entanglement-breaking. By definition of the weight-based measure,
\begin{align}
C_{w}\left(\mathcal{N}_{1} \otimes \mathcal{N}_{2}\right) &\leq 1-\left(1-s_{1}\right)\left(1-s_{2}\right),\\
&=s_{1}+s_{2}-s_{1} s_{2}.
\end{align}
Similarly, the same result holds for Eq.~\eqref{eq:otimes}.
\end{proof}

\begin{proposition}
The weight-based measure $C_w(\mathcal{N})$ can be computed with following conic program:
\begin{equation}\label{eq:3_4}
\begin{aligned}
C_w(\mathcal{N})= & \max 1-\operatorname{Tr}\left[W \mathcal{J}_{\mathcal{N}}\right] \\
\text {s.t. } & W^{\dagger}=W, \\
& \operatorname{Tr}[\mathcal{J}_{\mathcal{N}}W]\ge0,\\
& \operatorname{Tr}[\mathcal{J}_{\mathcal{M}}W]\ge1, \\
& \forall\mathcal{N}\in \operatorname{CPTP},\forall\mathcal{M}\in \operatorname{EB},
\end{aligned}
\end{equation}
where $\mathcal{J}_{\mathcal{N}}$ is bipartite Choi states and $\mathcal{J}_{\mathcal{M}}$ is separable Choi states.   
\end{proposition}
The proof can be found in Appendix~\ref{sec:appdix_fa}. 

Next, we establish a general lower bound for the weight-based measure of quantum memory. For the sake of clarity in the subsequent proof, we first introduce the definition of the robustness of quantum memory.

For any quantum channel $\mathcal{N}\in \operatorname{CPTP}(A_0\to A_1)$, the robustness of quantum memory is defined as~\cite{yuan2021universal}
\begin{align}\label{eq:ROM} 
C_r(\mathcal{N})=\min\{s\geq0: \mathcal{M}\in \operatorname{CPTP}(A_0\to A_1),\\\nonumber \frac{1}{1+s}(\mathcal{N}+s\mathcal{M})\in\operatorname{EB} \}. 
\end{align}

Equivalently, the robustness of quantum memory can be expressed in terms of the Choi states of these channels as~\cite{yuan2021universal}
\begin{align}\label{eq:ROM-state}
C_r(\mathcal{N})
=
\min\{s\geq0: \mathcal{J_M}\in \textrm{D}(A_0A_1),\\\nonumber \frac{1}{1+s}(\mathcal{J_N}+s\mathcal{J_M})\in\operatorname{SEP} \}, \end{align}
where $\operatorname{SEP}$ denotes the set of separable states.
It is worth noting that this quantity coincides with the robustness of entanglement for quantum states~\cite{vidal1999robustness}.

\begin{proposition}\label{prop:weight-robustness}
Given a channel $\mathcal{N}\in\operatorname{CPTP}(A_0\to A_1)$, its weight-based measure $C_w$ has the following lower bound:
\begin{equation}\label{eq:xiajie_1}
    C_w(\mathcal{N})\ge \frac{d_{A_0} \max \text{eig}(\mathcal{J_N}) - 1}{d_{A_0}^2-1},
\end{equation}
where $\text{eig}$ represents the eigenvalue of the Choi state $\mathcal{J_N}$. 
Furthermore, for a single-qubit channel $\mathcal{N}$, we have
\begin{equation}\label{eq:xiajie_2}
    C_w(\mathcal{N})\ge \frac{d_{A_0} \max \text{eig}(\mathcal{J_N}) - 1}{|\lambda_1\lambda_2\lambda_3|^{\frac{1}{8}}},
\end{equation}
where $\lambda_i(i=1,2,3)$ represents the eigenvalues of the matrix $T=S^T S$, and $S=(s_{ij})_{3\times3}$ is determined by the Choi state:
\begin{equation}
    \mathcal{J_N} = \frac{I}{4} + \sum_{i=1}^3 a_i \sigma_i \otimes I + \sum_{j=1}^3 b_j I \otimes \sigma_j \\
+ \frac{1}{4} \sum_{i,j=1}^3 s_{ij} \sigma_i \otimes \sigma_j.
\end{equation}
\end{proposition}
\begin{proof}
Let 
\[
\mathcal{J_N}=[1 - C_w(\mathcal{J_N})]\sigma^* + C_w(\mathcal{J_N})\tau^*
\]
be the optimal decomposition given in Eq.~(\ref{eq:eq2}), where $\sigma^*\in\textrm{D}(A_0A_1)$ is a separable state and $\tau^*\in\textrm{D}(A_0A_1)$ is a quantum state. 
Since the robustness of quantum memory $C_r(\mathcal{N})$ is a convex function~\cite{yuan2021universal}, we obtain
\begin{align}
C_r(\mathcal{N})
&\nonumber=
C_r([1 - C_w(\mathcal{J_N})]\sigma^* + C_w(\mathcal{J_N})\tau^*)
\\&\nonumber\le 
[1 - C_w(\mathcal{J_N})]C_r(\sigma^*) + C_w(\mathcal{J_N})C_r(\tau^*) 
\\&= 
C_w(\mathcal{J_N})C_r(\tau^*),
\end{align}
where the first equality follows from the fact that the robustness of quantum memory can equivalently be expressed in terms of Choi states, in direct analogy with the robustness of entanglement~\cite{yuan2021universal,vidal1999robustness}.
The inequality follows from the fact that $\sigma^*$ is a separable state and $C_r(\sigma^*)=0$. Additionally, $C_r(\tau^*)\le d_{A_0}^2-1$ implies that
\begin{align}\label{eq:35}
C_w(\mathcal{N})
&\nonumber\ge
\frac{C_r(\mathcal{J_N})}{C_r(\tau^*)} 
\\&\nonumber\ge
\frac{1}{d_{A_0}^2-1}C_r(\mathcal{J_N})
\\&\ge 
\frac{d_{A_0} \max \text{eig}(\mathcal{J_N}) - 1}{d_{A_0}^2-1}.
\end{align}
The third inequality holds is due to $C_r(\mathcal N) \ge d_{A_0} \max \text{eig}(\mathcal{J}_{\mathcal N}) - 1$ as described in Ref.~\cite{yuan2021universal}, which proves the Eq.~(\ref{eq:xiajie_1}).
For a single-qubit channel, according to a key result from Ref.~\cite{chang2024visually}, we have a compact upper bound
\begin{equation}\label{eq:36}
    C_r(\tau^*)\le |\lambda_1\lambda_2\lambda_3|^{\frac{1}{8}}.
\end{equation}
Thus, we can prove Eq.~(\ref{eq:xiajie_2}) by substituting Eq.~(\ref{eq:36}) into Eq.~(\ref{eq:35}).
\end{proof}

From Proposition~\ref{prop:weight-robustness}, we further obtain the following relation between the weight-based measure of quantum memory and the robustness of quantum memory:
\begin{equation}\label{eq:weight-robustness}
   C_w(\mathcal{N}) 
   \ge
   \frac{1}{d_{A_0}^2-1}C_r(\mathcal{N}).
\end{equation}
\subsection{The
performance advantage of quantum memory in nonlocal exclusion task}\label{sec:application}
In this section, we will discuss the advantage of quantum memory in exclusion tasks. 
We consider quantum games similar to Refs.~\cite{yuan2021universal,rosset2018resource,uola2020all}. Suppose there are two people Alice and Bob. Alice promises to select a quantum state $\sigma_i$ from a set of states $\{\sigma_i\}$ according to some probability distribution and encodes it in a memory $\mathcal{N}$, and sends it to Bob. Bob receives the state through the same quantum memory channel $\mathcal{N}$, and then performs a measurement using a set of observables $\{O_j\}$. When the input state is $\sigma_i$, the probability that Bob guesses $\sigma_j$ is $\operatorname{Tr}[\mathcal{N}(\sigma_i)O_j]$. For each pair of input label $i$ and measurement outcome $j$, Bob earns a payoff determined by a real-valued coefficient $\alpha_{i,j}$. The overall performance of the memory channel $\mathcal{N}$ in the game is quantified by the average payoff:
\begin{equation}
    P(\mathcal{N}, G)=\sum_{i, j} \alpha_{i, j} \operatorname{Tr}\left[\mathcal{N}\left(\sigma_{i}\right) O_{j}\right],
\end{equation}
where $\mathcal{G}=(\{\sigma_i\},\{O_j\},\{\alpha_{i,j}\})$ defines the game. This general framework captures a wide class of tasks, including quantum state discrimination and other decision-making scenarios, and is used to assess how well the memory channel preserves quantum information relative to classical (entanglement-breaking) channels.

By using Choi–Jamiołkowski isomorphism, we can get this result
\begin{equation}
\begin{aligned}
\mathcal{N}(\rho) &= \sum_{i,j=1}^d \rho_{ij} \, N\big( |i\rangle \langle j| \big),\\
&= d\operatorname{Tr}_A \left[ (\rho^{T} \otimes I_B) \mathcal{J}_\mathcal{N} \right],
\end{aligned}
\end{equation}
where $\mathcal{J}_\mathcal{N}$ is the Choi state of $\mathcal{N}$, $d$ is the dimension of the input system.

Then we can write
\begin{equation}
\mathcal{P}(\mathcal{N}, \mathcal{G})=d \sum_{i, j} \alpha_{i, j} \operatorname{Tr}\left[\mathcal{J}_\mathcal{N}\left(\sigma_{i}^{T} \otimes O_{j}\right)\right]=\operatorname{Tr}\left[\mathcal{J}_\mathcal{N} W\right],
\end{equation}
where $\sigma^{T}_{i}$ is the transpose of $\sigma_i$ and $W$ is a Hermitian operator
\begin{equation}
W=d \sum_{i, j} \alpha_{i, j} \sigma_{i}^{T} \otimes O_{j},
\end{equation}
with real coefficients $\alpha_{i,j}$.

Now, we connect the weight-based measure of quantum memory with its actual performance in a game task. The following proposition shows that the maximum advantage of quantum memory can be precisely quantified in this manner:
\begin{proposition}
Let $\mathcal{G}=(\{\sigma_i\},\{O_j\},\{\alpha_{i,j}\})$ denotes a particular game. Then the maximal pay-off of the game is 
\begin{equation}\label{eq:eq21}
\max _{\mathcal{G} \in \mathcal{S}_{G}} \mathcal{P}(\mathcal{N}, \mathcal{G})=1-C_w(\mathcal{N}),
\end{equation}
where the maximisation is over all games $\mathcal{G}\in \mathcal{S}_G$ with $\mathcal{S}_{G}=\{\mathcal{G}: \mathcal{P}(\mathcal{N}, \mathcal{G}) \geq 0, \mathcal{P}(\mathcal{M}, \mathcal{G}) \ge 1, \forall \mathcal{N} \in \operatorname{CPTP}, \mathcal{M} \in \textrm{EB}\}$.
\end{proposition}
\begin{proof}
We first write the weight-based measure in the standard form of the primal optimisation problem as
\begin{equation}
\begin{aligned}
C_w(\mathcal{N})= & \min 1-\operatorname{Tr}[x_1] \\
\text { s.t. } & x_1+x_2=\mathcal{J}_N , \\
& x_1\in \operatorname{cone}(\mathfrak{F}),x_2\in\operatorname{cone}(\mathcal{V}) .
\end{aligned}
\end{equation}
Now define $W = \operatorname{cone}(\mathcal{V})\oplus\operatorname{cone}(\mathcal{V})$, $\mathcal{W'}=\mathcal{V}$, $\mathcal{K}=\operatorname{cone}(\mathfrak{F})\oplus\operatorname{cone}(\mathcal{V})$, $\Lambda(x_1\oplus x_2)=x_1+x_2$, $A=I\oplus0$, $y=\mathcal{J}_{\mathcal{M}}$.
The dual form of the optimisation gives
\begin{equation}
\begin{aligned}
C_w(\mathcal{N})= & \max 1-\operatorname{Tr}[W\mathcal{J}_{\mathcal{N}}] \\
\text { s.t. } & \langle \Lambda^*(W)-I\oplus0,k\rangle\ge 0,\forall k\in \mathcal{K} ,  \\
& W\in \mathcal{W'} ,\\
\end{aligned}
\end{equation}
where 
\begin{equation}
\begin{aligned}
\left\langle \Lambda^*(W)-I\oplus0,k\right\rangle &\ge 0 ,\forall k \in \mathcal{K} 
\Longleftrightarrow\\
&\operatorname{Tr}[Wx_1]+\operatorname{Tr}[Wx_2]\ge\operatorname{Tr}[x_1],\\
&\forall x_1\in\operatorname{cone}(\mathfrak{F}),\forall x_2 \in\operatorname{cone}(\mathcal{V}).
\end{aligned} 
\end{equation}
Note that the above condition is further equivalent to 
\begin{equation}
\begin{aligned}
\left(\operatorname{Tr}\left[W x_{1}\right] \ge \operatorname{Tr}\left[x_{1}\right], \forall x_{1} \in \operatorname{cone}(\mathfrak{F})\right) \wedge \\ 
\left(\operatorname{Tr}\left[W x_{2}\right] \geq 0, \forall x_{2} \in \operatorname{cone}(\mathcal{V})\right).
\end{aligned}
\end{equation}
We can express it as 
\begin{equation}
\begin{aligned}
C_w(\mathcal{N})= & \max 1-\operatorname{Tr}\left[W \mathcal{J}_{\mathcal{N}}\right] \\
\text {s.t. } & W^{\dagger}=W, \\
& \operatorname{Tr}[\mathcal{J}_{\mathcal{M}}W]\ge1, \\
& \forall\mathcal{M}\in \textrm{EB},
\end{aligned}
\end{equation}
which is exactly the maximal pay-off function $\max _{\mathcal{G} \in \mathcal{S}} \mathcal{P}(\mathcal{N}, \mathcal{G})=1-C_w(\mathcal{N})$.
\end{proof}
Then the following proposition further explores this quantum advantage in the context of all nonlocal quantum games by considering the ratio of payoffs achievable with quantum memory versus classical memories.
\begin{proposition}
The advantage that a quantum memory $\mathcal{N}$ can provide over classical memories in all nonlocal quantum games is given by
\begin{equation}
\max _{\mathcal{G} \in \mathcal{S}^{\prime}} \frac{P(\mathcal{N}, \mathcal{G})}{\max _{\mathcal{M} \in \mathfrak{F}} \mathcal{P}(\mathcal{M}, \mathcal{G})}=1-C_w(\mathcal{N}),
\end{equation}
where $\mathcal{G} \in \mathcal{S}^{\prime} \text { with } \mathcal{S}^{\prime}=\{\mathcal{G}: \mathcal{P}(\mathcal{M}, \mathcal{G}) \geq 1, \forall \mathcal{M} \in \textrm{EB}\}$.
\end{proposition}
\section{Fundamental limitations on the performance of the identity channel distillation protocol}\label{sec:fundamental}
In this section we show another operational meaning based on a general bound on the error necessarily incurred in any transformation of a channel under free superchannels.
Suppose $\lambda=1-s$, then we can rewrite Eq.~\eqref{eq:eq2} as
\begin{equation}
C(\mathcal{E})=\min \left\{\lambda \mid \mathcal{J}_{\mathcal{E}} \geq \lambda \mathcal{J}_{\mathcal{M}}, \mathcal{M} \in \textrm{EB} \right\},
\end{equation}
which is similar to the measures of channels mentioned in Ref.~\cite{regula2021fundamental}.

Now, we define the overlap fidelity of a unitary channel $\mathcal{U}$, where the reference state $\psi$ is strictly defined on a bipartite system.
\begin{equation}
\begin{aligned}
F(\mathcal{U}): & =\max _{\mathcal{M} \in EB} F(\mathcal{U}, \mathcal{M}) ,\\
& =\max _{\mathcal{M} \in EB} \min _{\psi}\langle\mathrm{id} \otimes \mathcal{U}(\psi), \mathrm{id} \otimes \mathcal{M}(\psi)\rangle,
\end{aligned}
\end{equation}
where in the second line we used the fact that $id\otimes \mathcal{U}(\psi)$ is a pure state.

We present the following proposition, which establishes a lower bound on the error parameter $\epsilon$ based on the channel's capacity $C(\mathcal{E})$ and properties of a resourceful unitary channel $\mathcal{U}$.
\begin{proposition}
If there exists a free superchannel $\Theta$ such that $F(\Theta(\mathcal{E}),\mathcal{U})\ge1-\epsilon$ for some resourceful unitary channel $\mathcal{U}$, then $\epsilon$ is bounded as follows
\begin{align}
\varepsilon &\geq [1 - F(\mathcal{U})] C(\mathcal{E}) \label{eq:32a} ,
\end{align}
Furthermore, this implies a more general bound:
\begin{align}
\varepsilon \geq [1 - \frac{1}{d} ] C(\mathcal{E}) \label{eq:32b}.
\end{align}
\end{proposition}
The derivation of Eq.~(\ref{eq:32a}) is analogous to the Theorem 1 in Ref.~\cite{regula2021fundamental}.
The significance of the term $1-\frac{1}{d}$ can be understood by its connection to entanglement. 
Specifically, this term represents the geometric measure of entanglement for a maximally entangled state $E_{G}(\Psi^+)$. This value arises because the maximum possible fidelity between a maximally entangled state $\Psi^+$ and any separable state is exactly $\frac{1}{d}$~\cite{hedges2010efficient}:
\begin{equation}\label{eq:geometric-entanglement}
E_{G}(\Psi^+) 
= 1 - \max_{\sigma\in \operatorname{SEP}}F(\Psi^+,\sigma) 
= 1-\max_{\sigma \in \operatorname{SEP}}\langle\Psi^+|\sigma|\Psi^+ \rangle 
= 1 - \frac{1}{d},
\end{equation}
where $\Psi^+$ is maximally entangled state and $\sigma$ is separable state.
The proof can be found in Appendix~\ref{sec:appdix_fc}.

Thus, the bound in Eq.~\eqref{eq:32b} incorporates a quantity representing the maximum possible geometric entanglement in a $d$-dimensional bipartite system. This formulation suggests that the weight-based measure under investigation is intimately linked to geometric entanglement.

\section{Examples of the weighted-based measure of quantum memory}\label{sec:examples}
In this section, we will give some exact analytical expressions of the weighted-based measure of quantum memory.
\begin{figure}[htbp]
    \centering
    \includegraphics[width=0.5\textwidth]{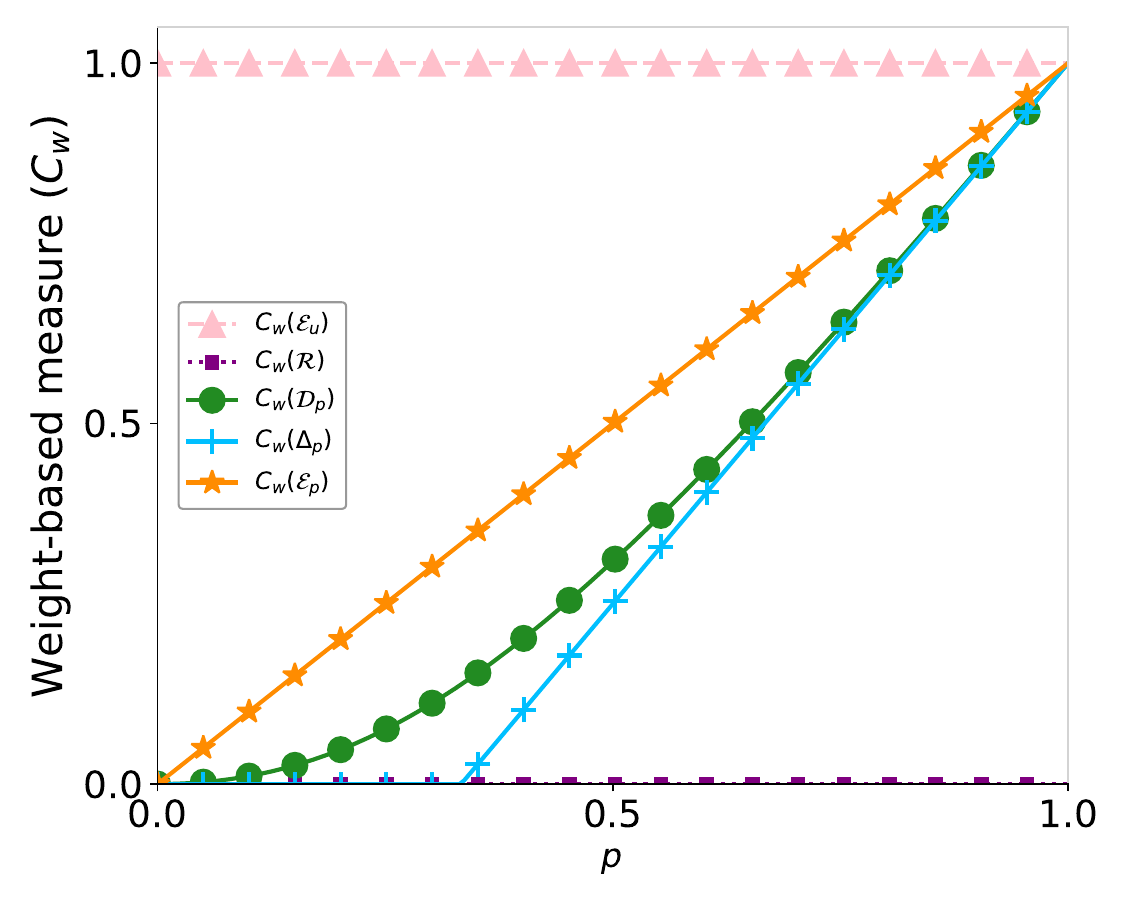} 
    \caption{To display and compare the variation of the weight-based measure of different quantum channels with respect to the parameter $p$.}
    \label{fig:chart}
\end{figure}
\subsection{Unitary channels}
A unitary channel describes a completely noiseless, ideal quantum evolution. This channel is reversible, loses no quantum information, and perfectly preserves quantum information~\cite{nielsen2010quantum}.

For the unitary channels $\mathcal{E}_u(\rho)=U\rho U^\dagger$, the corresponding Choi state is equivalent to Ref.~\cite{choi2023unital} that
\begin{equation}
\mathcal{J}_{\mathcal{E}_u} = (I\otimes U)\Psi^+(I\otimes U^\dagger)=\ket{\Psi_{\mathcal{U}}}\bra{\Psi_{\mathcal{U}}},
\end{equation}
where $\Psi^+$ is maximally entangled state and $\ket{\Psi_{\mathcal{E}_u}} = \frac{1}{\sqrt{d}}\sum_{i=0}^{d-1}\ket{i}\otimes(U\ket{i})$.

Let the Hermitian operator $W^*=\frac{d}{d-1}(I\otimes I-\ket{\Psi_{\mathcal{E}_u}}\bra{\Psi_{\mathcal{E}_u}})$. 
For any channel $\mathcal{N}$, we can check that
\begin{equation}
    \operatorname{Tr}[\mathcal{J}_{\mathcal{N}}W^*]\ge 0, \forall\mathcal{N}\in \operatorname{CPTP}.
\end{equation}
The inequality holds is due to
\begin{equation}
\begin{aligned}
\operatorname{Tr}[\mathcal{J}_{\mathcal{N}}W^*] &= \frac{d}{d-1}\operatorname{Tr}[\mathcal{J}_{\mathcal{N}}(I\otimes I-\ket{\Psi_{\mathcal{E}_u}}\bra{\Psi_{\mathcal{E}_u}})],\\
&=\frac{d}{d-1}(1-\bra{\Psi_{\mathcal{E}_u}}\mathcal{J}_{\mathcal{N}}\ket{\Psi_{\mathcal{E}_u}}),\\
&\ge  0.
\end{aligned}
\end{equation}

Similarly, for any $\operatorname{EB}$ channel $\mathcal{M}$, one can verify that 
\begin{equation}
    \operatorname{Tr}[\mathcal{J}_{\mathcal{M}}W^*]\ge 1, \forall\mathcal{N}\in \operatorname{EB},
\end{equation}
since $\bra{\Psi_{\mathcal{E}_u}}\mathcal{J}_{\mathcal{M}}\ket{\Psi_{\mathcal{E}_u}}\le\frac{1}{d}$ as established in Eq.~(\ref{eq:geometric-entanglement}).

Then $W^*$ satisfies the constraints in Eq.~\eqref{eq:3_4}. 
We have 
\begin{equation}
\begin{aligned}
C_w(\mathcal{E}_u) &= 1 - \min\operatorname{Tr}[\mathcal{J}_{\mathcal{E}_u}W^*], \\
&= 1 - \min \frac{d}{d-1}\operatorname{Tr}[(\ket{\Psi_{\mathcal{E}_u}}\bra{\Psi_{\mathcal{E}_u}})(I\otimes I-\ket{\Psi_{\mathcal{E}_u}}\bra{\Psi_{\mathcal{E}_u}})],\\
&=1- \frac{d}{d-1}\operatorname{Tr}[\ket{\Psi_{\mathcal{E}_u}}\bra{\Psi_{\mathcal{E}_u}}-1],\\
&=1.
\end{aligned}
\end{equation}
\subsection{Depolarising channels}\label{subsec:depolarising}
The depolarising channel is a symmetric noise model, often likened to quantum "white noise," where a quantum state is either left unchanged with probability $p$ or replaced by a completely random mixed state with probability $1-p$~\cite{nielsen2010quantum,wilde2013quantum}. 
For the depolarising channels 
\[
\Delta_{p}(\rho)=p \rho+(1-p) \frac{I}{2},
\]
the corresponding Choi state is the Werner state like Ref.~\cite{yuan2021universal} mentioned
\begin{equation}
\mathcal{J}_{\Delta_{p}}=p \Psi^{+}+(1-p) \frac{I}{4}.
\end{equation}
Here, we let Hermitian operator $W^*=2I\otimes I-2\ket{\Psi^+}\bra{\Psi^+}=2I - 2\Psi^+$, where $\ket{\Psi^+}=\frac{1}{\sqrt{2}}(\ket{00}+\ket{11})$.
For any channel $\mathcal{N}$, we can check that
\begin{equation}
    \operatorname{Tr}[\mathcal{J}_{\mathcal{N}}W^*]\ge 0, \forall\mathcal{N}\in \operatorname{CPTP},
\end{equation}
due to $1 - \operatorname{Tr}[\Psi^+\mathcal{J}_{\mathcal{N}}]\ge 0$.

Similarly, for any $\operatorname{EB}$ channel $\mathcal{M}$, we can check that
\begin{equation}
    \operatorname{Tr}[\mathcal{J}_{\mathcal{M}}W^*]\ge 1, \forall\mathcal{M}\in \operatorname{EB},
\end{equation}
which follows from Eq.~(\ref{eq:geometric-entanglement}). Then $W^*$ satisfies the constraints in Eq.~\eqref{eq:3_4}. 

We can show that $W^*$ is the optimal witness of $C_w(\Delta_p)$ by proving 
\begin{equation}
C_w\left(\Delta_{p}\right) \ge 1-\operatorname{Tr}[W^* \mathcal{J}_{\Delta_{p}}]=\frac{3 p-1}{2}.
\end{equation}
This is equivalent to finding a separable Choi state $\Phi^+_{\mathcal{M}}$ that satisfies
\begin{equation}
\mathcal{J}_{\Delta_{p}} \ge \frac{3(1-p)}{2} \Phi_{\mathcal{M}}^{+} .
\end{equation}
This is satisfied by choosing $\Phi^+_{\mathcal{M}}=\frac{1}{3}\Psi^++\frac{1}{6}I$.

Now we calculate the pay-off of this game is 
\begin{equation}
\begin{aligned}
\mathcal{P}(\Delta_p,\mathcal{G})&=\operatorname{Tr}[(2I - 2\Psi^+)(p \Psi^{+}+(1-p) I / 4)],\\
&=\frac{3-3p}{2}.
\end{aligned}
\end{equation}
This provides an upper bound
\begin{equation}
\begin{aligned}
C_w(\Delta_p) &\le \max\{ 1-\mathcal{P}(\Delta_p, \mathcal{G}), 0 \},\\
&=\begin{cases}
\frac{3p - 1}{2} & p \in [1/3, 1] \\
\ \ \ 0 & p \in [0, 1/3)
\end{cases}
\end{aligned}
\end{equation}
Therefore, the weight-based measure of the depolarising channels is
\begin{equation}
\begin{aligned}
C_w(\Delta_p) &=\begin{cases}
\frac{3p - 1}{2} & p \in [1/3, 1] \\
\ \ \ 0 & p \in [0, 1/3).
\end{cases}
\end{aligned}
\end{equation}

Since the Choi state of the depolarizing channel is a Werner state, our calculation for the channel simultaneously yields an independent derivation of the weight-based entanglement measure for the Werner state, thereby providing an alternative proof of the result of Lewenstein and Sanpera, who introduced it under the name \textit{best separable approximation}~\cite{lewenstein1998separability}.
\subsection{Maximal replacement channels}
A maximal replacement channel is characterized by its ability to completely erase the information of any input state $\rho$, replacing it with a predetermined maximally coherent state $\phi^+$~\cite{saxena2020dynamical,luo2025one}.

For the maximal replacement channels
 \begin{equation}
\mathcal{R}(\rho)=\operatorname{Tr}(\rho) \phi_{A_{1}}^{+},     
\end{equation} 
the corresponding Choi state is like Ref.~\cite{luo2022coherence} mentioned
\begin{equation}\label{maximal_choi_state}
\mathcal{J}_{\mathcal{R}}=\frac{I_{A_{0}}}{|A_0|} \otimes \phi_{A_{1}}^{+}
\end{equation}
where $\phi_{A_{1}}^{+}=\frac{1}{|A_1|}\sum_{i,j=1}^{|A_1|}\ket{i}\bra{j}$ is the maximally coherent state in dimension $d$.

Now, consider the Hermitian operator $W^*=2I_{A_0A_1}-|A_0|\mathcal{J}_{\mathcal{R}}=2I - (I_{A_0}\otimes\phi^+_{A_1})$.
For any channel $\mathcal{N}$, it can be verified that the following inequality holds:
\begin{equation}
    \operatorname{Tr}[\mathcal{J}_{\mathcal{N}}W^*]\ge 0, \forall\mathcal{N}\in \operatorname{CPTP}.
\end{equation}

By substituting the definition of $W^*$, the expression can be expanded as:

\begin{equation}
\begin{aligned}
\operatorname{Tr}[\mathcal{J}_{\mathcal{N}}W^*] &= \operatorname{Tr}[\mathcal{J}_{\mathcal{N}}(2I - |A_0|\mathcal{J_R}],\\
&=2\operatorname{Tr}[\mathcal{J}_{\mathcal{N}}] - |A_0|\operatorname{Tr}[\mathcal{J_R}\mathcal{J}_{\mathcal{N}}],\\
&=2\operatorname{Tr}[\mathcal{J}_{\mathcal{N}}] - \operatorname{Tr}[(I_{A_0}\otimes\phi^+_{A_1})\mathcal{J}_{\mathcal{N}}],\\
&\ge 0,
\end{aligned}
\end{equation}
where the inequality holds is due to 
$\operatorname{Tr}[AB]\le\lambda_{\max}(A)\operatorname{Tr}[B]$ holds with $A=I_{A_0}\otimes\phi^+_{A_1}$, $B=\mathcal{J}_{\mathcal{N}}$ and $\lambda_{\max}(A)$ denotes the maximal eigenvalue of matrix $A$.
The proof of this inequality is given in Appendix~\ref{sec:appdix_fd}.

For any $\operatorname{EB}$ channel $\mathcal{M}$, the operator $W^*$ satisfies the following condition: 
\begin{equation}
    \operatorname{Tr}[\mathcal{J}_{\mathcal{M}}W^*]\ge 1, \forall\mathcal{M}\in \operatorname{EB}.
\end{equation}
This property is a direct consequence of the separable structure of the Choi operator $\mathcal{J}_{\mathcal{M}}$ for any EB channel, which can be written as $\mathcal{J}_{\mathcal{M}}=\sum_kp_k\rho_k\otimes\sigma_k$. Substituting this form into the trace expression gives:

\begin{equation}
\begin{aligned}
\operatorname{Tr}[\mathcal{J}_{\mathcal{M}}W^*] &= \operatorname{Tr}[\mathcal{J}_{\mathcal{M}}(2I - (I_{A_0}\otimes\phi^+_{A_1}))],\\
&=2\operatorname{Tr}[\mathcal{J}_{\mathcal{M}}] - \operatorname{Tr}[(I_{A_0}\otimes\phi^+_{A_1})\mathcal{J}_{\mathcal{M}}],\\
&=2\operatorname{Tr}[\mathcal{J}_{\mathcal{M}}] - \sum_kp_k\operatorname{Tr[\rho_k]\operatorname{Tr}[\phi^+_{A_1}\sigma_k]},\\
&\ge 1.
\end{aligned}
\end{equation}

Since $\rho_k$ and $\sigma_k$ are quantum states, we know $\operatorname{Tr}(\rho_k)=1$ and the overlap term $\operatorname{Tr}[\phi^+_{A_1}\sigma_k]\le 1$.
Then $W^*$ satisfies the constraints in Eq.~\eqref{eq:3_4}. 
So we have 
\begin{equation}
\begin{aligned}
C_{w}(\mathcal{R}) & =1-\min\operatorname{Tr}[\mathcal{J}_{\mathcal{R}} W^*] ,\\
&= 1-(2\operatorname{Tr}[J_{\mathcal{R}}]-|A_0|\operatorname{Tr}[J_{\mathcal{R}}^{2}]),\\
&=0.
\end{aligned}
\end{equation}
Besides, we can directly see that the weight-based measure of the maximal replacement channels $C_{w}(\mathcal{R})=0$ by combining Eq.~(\ref{eq:choi_state}) with Eq.~(\ref{maximal_choi_state}).

The result that maximal replacement channels have zero channel memory can also be understood from their physical interpretation: These channels map any input state \(\rho\) to a maximally coherent state \(\phi^+_{A_1}\), meaning that all original information is lost through the channel. Indeed, by comparing with Eq.~(\ref{eq:EB_channels}), one can see that maximal replacement channels are a subset of \(\textrm{EB}\) channels.

\subsection{Stochastic damping channels}
We consider the stochastic damping channel $\mathcal{D}_p$ defined as follows, which perfectly transmits a quantum state with probability p and deterministically projects it onto the ground state $\ket{0}$ with the complementary probability $1-p$. This is a non-unitary noise model that describes a probabilistic "reset" process for a qubit~\cite{nielsen2010quantum}.

For the stochastic damping channels $\mathcal{D}_p(\rho)=p\rho+(1-p)\ket{0}\bra{0}$, the corresponding Choi state is
\begin{equation}
\mathcal{J}_{\mathcal{D}_{p}}=p \Psi^{+}+(1-p) \frac{I}{2} \otimes\ket{0}\bra{0} .
\end{equation}
and the matrix form of $\mathcal{J}_{\mathcal{D}_{p}}$ is
\begin{equation}
\mathcal{J}_{\mathcal{D}_{p}}=\left(\begin{array}{cccc}
\frac{1}{2} & 0 & 0 & \frac{p}{2} \\
0 & 0 & 0 & 0 \\
0 & 0 & \frac{1-p}{2} & 0 \\
\frac{p}{2} & 0 & 0 & \frac{p}{2}
\end{array}\right).
\end{equation}
Suppose the spectral decomposition of this matrix is 
\begin{equation}
\mathcal{J}_{\mathcal{D}_{p}}=\lambda_{0} \psi_{0}+\lambda_{1} \psi_{1}+\lambda_{2} \psi_{2}+\lambda_{3} \psi_{3},
\end{equation}
where $\psi_i=\ket{\phi_i}\bra{\phi_i}$ represent the density matrices of the corresponding pure eigenstates $\ket{\phi_i}$ and the eigenvalues are
\begin{align}
\lambda_0 &= \frac{1 + p + \sqrt{1 - 2p + 5p^2}}{4} ,\\
\lambda_1 &= \frac{1 + p - \sqrt{1 - 2p + 5p^2}}{4} ,\\
\lambda_2 &= \frac{1 - p}{2} ,\\
\lambda_3 &= 0.
\end{align}
Now, we suppose the game with the Hermitian operator $W^*=\frac{d}{d-1}(I\otimes I-\psi_0)=2I-2\psi_0$. 

We can check $W^*$ satisfies the first constraints in Eq.~\eqref{eq:3_4} because of $0\le \operatorname{Tr}[\mathcal{J}_{\mathcal{N}}\psi_0]\le 1$. As shown in Ref.~\cite{yuan2021universal}, we have the lemma that the maximum overlap between a maximally entangled state and any separable state is $\frac{1}{d}$, in which $d$ is the input dimension. Then we can prove that $\operatorname{Tr}[\mathcal{J}_{\mathcal{N}}W^*]=2(2-\operatorname{Tr}[\mathcal{J}_{\mathcal{N}}\psi_0])\ge1$. The above proves that $W^*$ satisfies all the constraints of Eq.~\eqref{eq:3_4}.

We want to calculate $\operatorname{Tr}(W^*\mathcal{J}_{\mathcal{D}_{p}})$, where $W^* = 2I - 2\psi_0$, and the spectral decomposition of $\mathcal{J}_{\mathcal{D}_{p}}$ is $\lambda_0\psi_0 + \lambda_1\psi_1 + \lambda_2\psi_2 + \lambda_3\psi_3$.

Then we substitute the expression
\begin{equation}
\operatorname{Tr}(W^*\mathcal{J}_{\mathcal{D}_{p}}) = \operatorname{Tr}((2I - 2\psi_0)(\lambda_0\psi_0 + \lambda_1\psi_1 + \lambda_2\psi_2 + \lambda_3\psi_3)).
\end{equation}
Expanding using the linearity of the trace:
\begin{align}
\operatorname{Tr}(W^*\mathcal{J}_{\mathcal{D}_{p}}) &= \sum_{i=0}^{3} \lambda_i \operatorname{Tr}[(2I - 2\psi_0)\psi_i], \\
&= \sum_{i=0}^{3} \lambda_i [2\operatorname{Tr}(I\psi_i) - 2\operatorname{Tr}(\psi_0\psi_i)].
\end{align}
Since $\psi_i$ is the density matrix of an eigenstate, $\operatorname{Tr}(\psi_i) = 1$, therefore $\operatorname{Tr}(I\psi_i) = \operatorname{Tr}(\psi_i) = 1$, thus we have
\begin{equation}
\operatorname{Tr}(W^*\mathcal{J}_{\mathcal{D}_{p}}) = \sum_{i=0}^{3} \lambda_i [2 \cdot 1 - 2\operatorname{Tr}(\psi_0\psi_i)] = 2 \sum_{i=0}^{3} \lambda_i (1 - \operatorname{Tr}(\psi_0\psi_i))
\end{equation}
Then we use the orthogonality of eigenstates: $\psi_i$ are density matrices of eigenstates corresponding to different eigenvalues, and they are mutually orthogonal. Therefore, we can get 
\begin{equation}
\operatorname{Tr}(W^*\mathcal{J}_{\mathcal{D}_{p}}) = 2(1 - \lambda_0).
\end{equation}
And finally we obtain the upper bound 
\begin{equation}
C_w(\mathcal{D}_p)\le1-\operatorname{Tr}(W^*\mathcal{J}_{\mathcal{D}_{p}})=2\lambda_0-1=\frac{\sqrt{1-2p+5p^2}+p-1}{2}.
\end{equation}

\subsection{Erasure channels}
The erasure channel is a key noise model in quantum information theory, in this channel, a quantum state is either transmitted without error or it is replaced by a known erasure state, which is orthogonal to the computational basis. This property of knowing exactly where an error has occurred greatly simplifies the process of error correction~\cite{bennett1997capacities,nielsen2010quantum}.

For the erasure channels $\mathcal{E}_{p}(\rho)=p \rho+(1-p)\ket{2}\bra{2}$ with $\ket{2}$ orthogonal to $\{\ket{0},\ket{1}\}$, the Choi state is 
\begin{equation}
\mathcal{J}_{\mathcal{E}_p} = p\Psi^++(1-p)\frac{I}{2}\otimes\ket{2}\bra{2}.
\end{equation}
Now, we suppose the game with the Hermitian operator $W^*=2I-2\Psi^+-I\otimes\ket{2}\bra{2}$. 

We can check $\operatorname{Tr}[\mathcal{J}_{\mathcal{N}}W^*]\ge0$ due to the trace of the product of two positive semidefinite operators is non-negative. For all separable states $\sigma$, consider a projective measurement $\{P_{01}=\ket{0}\bra{0}+\ket{1}\bra{1},P_2=\ket{2}\bra{2}\}$ on the second system, then
\begin{equation}
\begin{aligned}
\operatorname{Tr}[W^* \sigma] &= \operatorname{Tr}[\left(P_{01}+P_{2}\right) W^*\left(P_{01}+P_{2}\right) \sigma] ,\\
&= \operatorname{Tr}[(2P_{01}-2\Psi^+)\sigma] + \operatorname{Tr}[P_2\sigma] ,\\
&\quad -2\operatorname{Tr}[P_{01}\Psi^+P_{01}\sigma]+\operatorname{Tr}[P_2\sigma P_2] ,\\
&= 2p_{01}-2p_{01}\operatorname{Tr}[\Psi^+\sigma_{01}]+p_2 ,\\
&= 2p_{01}(1-\operatorname{Tr}[\Psi^+\sigma_{01}])+p_2 ,\\
&\ge 1.
\end{aligned}
\end{equation}
Here $p_{01}=\operatorname{Tr}[P_{01}\sigma_{01} P_{01}], p_2=\operatorname{Tr}[P_2\sigma_2P_2], \sigma_{01}=P_{01}\sigma P_{01}/p_{01}$. The last line follows from $2\operatorname{Tr}[\Psi^+\sigma_{01}]\le 1$ and $p_{01}+p_2=1$. Thus, the above proves that $W^*$ satisfies all the constraints of Eq.~\eqref{eq:3_4}.
Then we need to compute
\begin{equation}
\begin{aligned}
\operatorname{Tr}[W^*\mathcal{J}_{\mathcal{E}_p}]
&= \operatorname{Tr}\Big[ \Big(2I - 2\Psi^+ - I\otimes\ket{2}\bra{2}\Big) \\
&\qquad \qquad \cdot \Big(p\Psi^+ + (1-p)\frac{I}{2}\otimes\ket{2}\bra{2}\Big) \Big] \\
&= 2p + 2 - 2p - 2p + p - 1 \\
&= 1-p.
\end{aligned}
\end{equation}
Finally we obtain the calculation
\begin{equation}
C_w(\mathcal{E}_p)\le1-\operatorname{Tr}[W^*\mathcal{J}_{\mathcal{E}_p}]=p.
\end{equation}
We can also prove that the estimate from the game witness is tight, which is equivalent to finding a separable Choi state $\Phi^+_{\mathcal{M}}$ that satisfies
\begin{equation}
\mathcal{J}_{\mathcal{E}_{p}} \ge (1-p) \Phi_{\mathcal{M}}^{+}.
\end{equation}
This is satisfied by choosing $\Phi^+_{\mathcal{M}}=\frac{I}{2}\otimes\ket{2}\bra{2}$. Therefore, the weight-based measure of the erasure channels is $p$.

We find that the weight-based measure for the depolarizing channels, stochastic damping channels, and erasure channels is numerically equal to their corresponding robustness measures~\cite{yuan2021universal}. The weight-based measure of these relative channels is shown in Fig.~\ref{fig:chart}.

\section{Discussion and Conclusion}\label{sec:conclusion}
In this paper, we proposed a quantification approach for quantum memory based on weight, along with its operational interpretation. We first introduced the definition of the weight-based of a quantum memory, which involves optimization over all input states, and proved that this optimization can be achieved via the memory’s Choi matrix. Consequently, the problem of computing the channel weight can be reduced to a SDP. We then studied some properties such as convexity and monotonicity, and provided an upper bound for the weight-based measure of two channels under the tensor product. 
Then, we establish a general lower bound for the weight-based measure of quantum memory and extend to a tighter lower bound for single-qubit memory. Besides, we further derive the relation between the weight-based measure of quantum memory $C_w$ and the robustness of quantum memory $C_r$. By introducing an nonlocal exclusion task, we demonstrated the operational advantage of measure in this setting. We examine the unavoidable lower bounds on the error incurred when distilling or purifying an arbitrary noisy channel into an ideal unitary channel or a replacement channel under the constraints of free superchannels, thereby revealing the close connection between the weight-based quantifier of quantum memories and the performance of these purification tasks. Finally, we calculated some special memory with weight-based measure. 
Besides, by calculating the measure for the depolarising channel, we provide an alternative proof for the best separable approximation of the Werner state~\cite{lewenstein1998separability,Regula2017convex}. This result is in agreement with the established separability criterion for Werner states, as demonstrated by Lewenstein and Sanpera~\cite{lewenstein1998separability}.

Furthermore, we observe an interesting phenomenon when comparing the quantum memory measure with the quantum capacity. For the the depolarizing channel, as an example in Sec.~\ref{subsec:depolarising}, there exists a parameter range of $\frac{1}{3}<p\le\frac{2}{3}$, where the quantum memory is positive 
($C_w(\Delta_p)>0$), yet the quantum capacity $\mathrm{Q}$ is known to be zero~\cite{bruss1998optimal,bennett1997capacities}.
As shown in Table \ref{tab:q_vs_cw_comparison}, there is a regime where the quantum capacity is zero while the quantum memory measure is positive.
According to their respective definitions, these are two different measures. The quantum memory measure quantifies the channel's ability to resist noise while preserving its quantum nature. In contrast, the quantum capacity defines the maximum rate at which quantum information can be transmitted reliably after applying quantum error correction. We conjecture that this phenomenon may suggest that in this range, the channel retains its quantum nature but fails to meet the standard required for effective quantum transmission.
\begin{table}[h!]
\centering
\begin{tabular}{l@{\hspace{2cm}}c@{\hspace{2cm}}c}
\toprule
Parameter Range   & $\mathrm{Q}$ & $C_w$ \\
\midrule
$p \le \frac{1}{3}$ & $0$ & $0$ \\
\\ 
$\frac{1}{3} < p \le \frac{2}{3}$ & $0$ & $> 0$ \\
\\ 
$\frac{2}{3} < p$ & $> 0$ & $> 0$ \\
\bottomrule
\end{tabular}
\caption{Comparison of Quantum Capacity ($\mathrm{Q}$) and Quantum Memory Measure ($C_w$) for the Depolarizing Channel.}\label{tab:q_vs_cw_comparison}
\end{table}

Although we have not established a direct relationship between the quantum memory measure and quantum capacity, this presents an open problem worthy of exploration. We propose this as a valuable direction for future work and hope the findings presented in this paper will contribute to a deeper understanding of this topic.

\begin{acknowledgments}
This work is supported by the National Natural Science Foundation of China Grant No. 62001274 and No. 62171266 and the Fundamental Research Funds for the Central Universities (GK202501008).
\end{acknowledgments}

\sloppy
\makeatletter
\makeatother
\bibliographystyle{unsrt}

\bibliography{main}

@article{nicolas2014quantum,
  title={A quantum memory for orbital angular momentum photonic qubits},
  author={Nicolas, A and Veissier, L and Giner, L and Giacobino, E and Maxein, D and Laurat, J},
  journal={Nature Photonics},
  volume={8},
  number={3},
  pages={234--238},
  year={2014},
  publisher={Nature Publishing Group UK London}
}

@article{freer2017single,
  title={A single-atom quantum memory in silicon},
  author={Freer, Solomon and Simmons, Stephanie and Laucht, Arne and Muhonen, Juha T and Dehollain, Juan P and Kalra, Rachpon and Mohiyaddin, Fahd A and Hudson, Fay E and Itoh, Kohei M and McCallum, Jeffrey C and others},
  journal={Quantum Science and Technology},
  volume={2},
  number={1},
  pages={015009},
  year={2017},
  publisher={IOP Publishing}
}

@article{wang2019efficient,
  title={Efficient quantum memory for single-photon polarization qubits},
  author={Wang, Yunfei and Li, Jianfeng and Zhang, Shanchao and Su, Keyu and Zhou, Yiru and Liao, Kaiyu and Du, Shengwang and Yan, Hui and Zhu, Shi-Liang},
  journal={Nature Photonics},
  volume={13},
  number={5},
  pages={346--351},
  year={2019},
  publisher={Nature Publishing Group UK London}
}

@article{hosseini2011high,
  title={High efficiency coherent optical memory with warm rubidium vapour},
  author={Hosseini, Mahdi and Sparkes, Ben M and Campbell, Geoff and Lam, Ping K and Buchler, Ben C},
  journal={Nature Communications},
  volume={2},
  number={1},
  pages={174},
  year={2011},
  publisher={Nature Publishing Group UK London}
}

@article{hsiao2018highly,
  title={Highly efficient coherent optical memory based on electromagnetically induced transparency},
  author={Hsiao, Ya-Fen and Tsai, Pin-Ju and Chen, Hung-Shiue and Lin, Sheng-Xiang and Hung, Chih-Chiao and Lee, Chih-Hsi and Chen, Yi-Hsin and Chen, Yong-Fan and Yu, Ite A and Chen, Ying-Cheng},
  journal={Physical Review Letters},
  volume={120},
  number={18},
  pages={183602},
  year={2018},
  publisher={APS}
}

@article{vernaz2018highly,
  title={Highly-efficient quantum memory for polarization qubits in a spatially-multiplexed cold atomic ensemble},
  author={Vernaz-Gris, Pierre and Huang, Kun and Cao, Mingtao and Sheremet, Alexandra S and Laurat, Julien},
  journal={Nature Communications},
  volume={9},
  number={1},
  pages={363},
  year={2018},
  publisher={Nature Publishing Group UK London}
}

@article{lvovsky2009optical,
  title={Optical quantum memory},
  author={Lvovsky, Alexander I and Sanders, Barry C and Tittel, Wolfgang},
  journal={Nature Photonics},
  volume={3},
  number={12},
  pages={706--714},
  year={2009},
  publisher={Nature Publishing Group UK London}
}

@article{goronkin2004high,
  title={High-performance emerging solid-state memory technologies},
  author={Goronkin, Herb and Yang, Yang},
  journal={MRS bulletin},
  volume={29},
  number={11},
  pages={805--813},
  year={2004},
  publisher={Cambridge University Press}
}

@article{vieira2024entanglement,
  title={Entanglement-breaking channels are a quantum memory resource},
  author={Vieira, Lucas B and Ku, Huan-Yu and Budroni, Costantino},
  journal={arXiv preprint arXiv:2402.16789},
  year={2024}
}

@article{simnacher2019certifying,
  title={Certifying quantum memories with coherence},
  author={Simnacher, Timo and Wyderka, Nikolai and Spee, Cornelia and Yu, Xiao-Dong and G{\"u}hne, Otfried},
  journal={Physical Review A},
  volume={99},
  number={6},
  pages={062319},
  year={2019},
  publisher={APS}
}

@article{yuan2021universal,
  title={Universal and operational benchmarking of quantum memories},
  author={Yuan, Xiao and Liu, Yunchao and Zhao, Qi and Regula, Bartosz and Thompson, Jayne and Gu, Mile},
  journal={npj Quantum Information},
  volume={7},
  number={1},
  pages={108},
  year={2021},
  publisher={Nature Publishing Group UK London}
}

@article{chang2024visually,
  title={Visually quantifying single-qubit quantum memory},
  author={Chang, Wan-Guan and Ju, Chia-Yi and Chen, Guang-Yin and Chen, Yueh-Nan and Ku, Huan-Yu},
  journal={Physical Review Research},
  volume={6},
  number={2},
  pages={023035},
  year={2024},
  publisher={APS}
}

@book{nielsen2010quantum,
  title={Quantum Computation and Quantum Information},
  author={Nielsen, Michael A and Chuang, Isaac L},
  year={2010},
  publisher={Cambridge University Press}
}

@article{wootters1982single,
  title={A single quantum cannot be cloned},
  author={Wootters, William K and Zurek, Wojciech H},
  journal={Nature},
  volume={299},
  number={5886},
  pages={802--803},
  year={1982},
  publisher={Nature Publishing Group UK London}
}

@article{heshami2016quantum,
  title={Quantum memories: emerging applications and recent advances},
  author={Heshami, Khabat and England, Duncan G and Humphreys, Peter C and Bustard, Philip J and Acosta, Victor M and Nunn, Joshua and Sussman, Benjamin J},
  journal={Journal of Modern Optics},
  volume={63},
  number={20},
  pages={2005--2028},
  year={2016},
  publisher={Taylor \& Francis}
}

@article{ku2022quantifying,
  title={Quantifying quantumness of channels without entanglement},
  author={Ku, Huan-Yu and Kadlec, Josef and {\v{C}}ernoch, Anton{\'\i}n and Quintino, Marco T{\'u}lio and Zhou, Wenbin and Lemr, Karel and Lambert, Neill and Miranowicz, Adam and Chen, Shin-Liang and Nori, Franco and others},
  journal={PRX Quantum},
  volume={3},
  number={2},
  pages={020338},
  year={2022},
  publisher={APS}
}

@article{uola2020all,
  title={All quantum resources provide an advantage in exclusion tasks},
  author={Uola, Roope and Bullock, Tom and Kraft, Tristan and Pellonp{\"a}{\"a}, Juha-Pekka and Brunner, Nicolas},
  journal={Physical Review Letters},
  volume={125},
  number={11},
  pages={110402},
  year={2020},
  publisher={APS}
}

@article{ducuara2020operational,
  title={Operational interpretation of weight-based resource quantifiers in convex quantum resource theories},
  author={Ducuara, Andr{\'e}s F and Skrzypczyk, Paul},
  journal={Physical Review Letters},
  volume={125},
  number={11},
  pages={110401},
  year={2020},
  publisher={APS}
}

@article{streltsov2010linking,
  title={Linking a distance measure of entanglement to its convex roof},
  author={Streltsov, Alexander and Kampermann, Hermann and Bru{\ss}, Dagmar},
  journal={New Journal of Physics},
  volume={12},
  number={12},
  pages={123004},
  year={2010},
  publisher={IOP Publishing}
}

@article{bu2018asymmetry,
  title={Asymmetry and coherence weight of quantum states},
  author={Bu, Kaifeng and Anand, Namit and Singh, Uttam},
  journal={Physical Review A},
  volume={97},
  number={3},
  pages={032342},
  year={2018},
  publisher={APS}
}

@article{horodecki2003entanglement,
  title={Entanglement breaking channels},
  author={Horodecki, Michael and Shor, Peter W and Ruskai, Mary Beth},
  journal={Reviews in Mathematical Physics},
  volume={15},
  number={06},
  pages={629--641},
  year={2003},
  publisher={World Scientific}
}

@article{tabia2024super,
  title={Super-activating quantum memory by entanglement-breaking channels},
  author={Tabia, Gelo Noel M and Hsieh, Chung-Yun},
  journal={arXiv preprint arXiv:2410.13499},
  year={2024}
}

@article{buscemi2011entanglement,
  title={Entanglement cost in practical scenarios},
  author={Buscemi, Francesco and Datta, Nilanjana},
  journal={Physical Review Letters},
  volume={106},
  number={13},
  pages={130503},
  year={2011},
  publisher={APS}
}

@article{brassard2005quantum,
  title={Quantum pseudo-telepathy},
  author={Brassard, Gilles and Broadbent, Anne and Tapp, Alain},
  journal={Foundations of Physics},
  volume={35},
  number={11},
  pages={1877--1907},
  year={2005},
  publisher={Springer}
}

@article{gottesman2001quantum,
  title={Quantum digital signatures},
  author={Gottesman, Daniel and Chuang, Isaac},
  journal={arXiv preprint quant-ph/0105032},
  year={2001}
}

@article{bennett2014quantum,
  title={Quantum cryptography: Public key distribution and coin tossing},
  author={Bennett, Charles H and Brassard, Gilles},
  journal={Theoretical Computer Science},
  volume={560},
  pages={7--11},
  year={2014},
  publisher={Elsevier}
}

@article{pirandola2020advances,
  title={Advances in quantum cryptography},
  author={Pirandola, Stefano and Andersen, Ulrik L and Banchi, Leonardo and Berta, Mario and Bunandar, Darius and Colbeck, Roger and Englund, Dirk and Gehring, Tobias and Lupo, Cosmo and Ottaviani, Carlo and others},
  journal={Advances in Optics and Photonics},
  volume={12},
  number={4},
  pages={1012--1236},
  year={2020},
  publisher={Optical Society of America}
}

@ARTICLE{10236453,
  author={Piveteau, Christophe and Sutter, David},
  journal={IEEE Transactions on Information Theory}, 
  title={Circuit Knitting With Classical Communication}, 
  year={2024},
  volume={70},
  number={4},
  pages={2734-2745},
  doi={10.1109/TIT.2023.3310797}
}

@ARTICLE{11131294,
  author={Brenner, Lukas and Piveteau, Christophe and Sutter, David},
  journal={IEEE Transactions on Information Theory}, 
  title={Optimal Wire Cutting With Classical Communication}, 
  year={2025},
  volume={71},
  number={10},
  pages={7742-7752},
  doi={10.1109/TIT.2025.3601047}
}

@article{jing2025circuit,
  title={Circuit knitting facing exponential sampling-overhead scaling bounded by entanglement cost},
  author={Jing, Mingrui and Zhu, Chengkai and Wang, Xin},
  journal={Physical Review A},
  volume={111},
  number={1},
  pages={012433},
  year={2025},
  publisher={APS}
}

@article{piveteau2025circuit,
  title={Circuit cutting with classical side information},
  author={Piveteau, Christophe and Schmitt, Lukas and Sutter, David},
  journal={Physical Review Research},
  volume={7},
  number={3},
  pages={033063},
  year={2025},
  publisher={APS}
}

@article{piveteau2025simulating,
  title={Simulating quantum circuits with restricted quantum computers},
  author={Piveteau, Christophe},
  journal={arXiv preprint arXiv:2503.21773},
  year={2025}
}

@article{gour2019comparison,
  title={Comparison of quantum channels by superchannels},
  author={Gour, Gilad},
  journal={IEEE Transactions on Information Theory},
  volume={65},
  number={9},
  pages={5880--5904},
  year={2019},
  publisher={IEEE}
}

@article{luo2025one,
  title={One-shot manipulation of coherence in dynamic quantum resource theory},
  author={Luo, Yu},
  journal={Physical Review A},
  volume={111},
  number={2},
  pages={022447},
  year={2025},
  publisher={APS}
}

@article{chitambar2019quantum,
  title={Quantum resource theories},
  author={Chitambar, Eric and Gour, Gilad},
  journal={Reviews of Modern Physics},
  volume={91},
  number={2},
  pages={025001},
  year={2019},
  publisher={APS}
}

@article{luo2022coherence,
  title={Coherence weight of quantum channels},
  author={Luo, Yu and Ye, Mingfei and Li, Yongming},
  journal={Physica A: Statistical Mechanics and its Applications},
  volume={599},
  pages={127510},
  year={2022},
  publisher={Elsevier}
}

@article{saxena2020dynamical,
  title={Dynamical resource theory of quantum coherence},
  author={Saxena, Gaurav and Chitambar, Eric and Gour, Gilad},
  journal={Physical Review Research},
  volume={2},
  number={2},
  pages={023298},
  year={2020},
  publisher={APS}
}

@article{liu2020operational,
  title={Operational resource theory of quantum channels},
  author={Liu, Yunchao and Yuan, Xiao},
  journal={Physical Review Research},
  volume={2},
  number={1},
  pages={012035},
  year={2020},
  publisher={APS}
}

@article{luo2024epsilon,
  title={Epsilon measures of state-based quantum resource theory},
  author={Luo, Yu and Meng, Fanxu and Wang, Youle},
  journal={Physical Review A},
  volume={109},
  number={5},
  pages={052413},
  year={2024},
  publisher={APS}
}

@article{liu2019resource,
  title={Resource theories of quantum channels and the universal role of resource erasure},
  author={Liu, Zi-Wen and Winter, Andreas},
  journal={arXiv preprint arXiv:1904.04201},
  year={2019}
}

@article{wilde2020amortized,
  title={Amortized channel divergence for asymptotic quantum channel discrimination},
  author={Wilde, Mark M and Berta, Mario and Hirche, Christoph and Kaur, Eneet},
  journal={Letters in Mathematical Physics},
  volume={110},
  pages={2277--2336},
  year={2020},
  publisher={Springer}
}

@article{duan2001long,
  title={Long-distance quantum communication with atomic ensembles and linear optics},
  author={Duan, L-M and Lukin, Mikhail D and Cirac, J Ignacio and Zoller, Peter},
  journal={Nature},
  volume={414},
  number={6862},
  pages={413--418},
  year={2001},
  publisher={Nature Publishing Group UK London}
}

@article{watrous2009semidefinite,
  title={Semidefinite programs for completely bounded norms},
  author={Watrous, John},
  journal={arXiv preprint arXiv:0901.4709},
  year={2009}
}

@article{zhong2015optically,
  title={Optically addressable nuclear spins in a solid with a six-hour coherence time},
  author={Zhong, Manjin and Hedges, Morgan P and Ahlefeldt, Rose L and Bartholomew, John G and Beavan, Sarah E and Wittig, Sven M and Longdell, Jevon J and Sellars, Matthew J},
  journal={Nature},
  volume={517},
  number={7533},
  pages={177--180},
  year={2015},
  publisher={Nature Publishing Group UK London}
}

@article{ye2023quantifying,
  title={Quantifying channel coherence via the norm distance},
  author={Ye, Mingfei and Luo, Yu and Li, Yongming},
  journal={Journal of Physics A: Mathematical and Theoretical},
  volume={57},
  number={1},
  pages={015307},
  year={2023},
  publisher={IOP Publishing}
}

@article{vidal1999robustness,
  title={Robustness of entanglement},
  author={Vidal, Guifr{\'e} and Tarrach, Rolf},
  journal={Physical Review A},
  volume={59},
  number={1},
  pages={141},
  year={1999},
  publisher={APS}
}

@article{wang2017single,
  title={Single-qubit quantum memory exceeding ten-minute coherence time},
  author={Wang, Ye and Um, Mark and Zhang, Junhua and An, Shuoming and Lyu, Ming and Zhang, Jing-Ning and Duan, L-M and Yum, Dahyun and Kim, Kihwan},
  journal={Nature Photonics},
  volume={11},
  number={10},
  pages={646--650},
  year={2017},
  publisher={Nature Publishing Group UK London}
}

@article{julsgaard2004experimental,
  title={Experimental demonstration of quantum memory for light},
  author={Julsgaard, Brian and Sherson, Jacob and Cirac, J Ignacio and Fiur{\'a}{\v{s}}ek, Jarom{\'\i}r and Polzik, Eugene S},
  journal={Nature},
  volume={432},
  number={7016},
  pages={482--486},
  year={2004},
  publisher={Nature Publishing Group UK London}
}

@article{choi2008mapping,
  title={Mapping photonic entanglement into and out of a quantum memory},
  author={Choi, Kyung Soo and Deng, Hui and Laurat, Julien and Kimble, HJ},
  journal={Nature},
  volume={452},
  number={7183},
  pages={67--71},
  year={2008},
  publisher={Nature Publishing Group UK London}
}

@article{zhao2009long,
  title={Long-lived quantum memory},
  author={Zhao, R and Dudin, YO and Jenkins, SD and Campbell, CJ and Matsukevich, DN and Kennedy, TAB and Kuzmich, A},
  journal={Nature Physics},
  volume={5},
  number={2},
  pages={100--104},
  year={2009},
  publisher={Nature Publishing Group UK London}
}

@article{jensen2011quantum,
  title={Quantum memory for entangled continuous-variable states},
  author={Jensen, Kasper and Wasilewski, Wojciech and Krauter, Hanna and Fernholz, Thomas and Nielsen, Bo Melholt and Owari, M and Plenio, Martin B and Serafini, A and Wolf, MM and Polzik, ES},
  journal={Nature Physics},
  volume={7},
  number={1},
  pages={13--16},
  year={2011},
  publisher={Nature Publishing Group UK London}
}

@article{pu2017experimental,
  title={Experimental realization of a multiplexed quantum memory with 225 individually accessible memory cells},
  author={Pu, YF and Jiang, Nan and Chang, Wei and Yang, HX and Li, Chang and Duan, LM},
  journal={Nature communications},
  volume={8},
  number={1},
  pages={15359},
  year={2017},
  publisher={Nature Publishing Group UK London}
}

@article{lei2023quantum,
  title={Quantum optical memory for entanglement distribution},
  author={Lei, Yisheng and Kimiaee Asadi, Faezeh and Zhong, Tian and Kuzmich, Alex and Simon, Christoph and Hosseini, Mahdi},
  journal={Optica},
  volume={10},
  number={11},
  pages={1511--1528},
  year={2023},
  publisher={Optica Publishing Group}
}

@article{terhal2015quantum,
  title={Quantum error correction for quantum memories},
  author={Terhal, Barbara M},
  journal={Reviews of Modern Physics},
  volume={87},
  number={2},
  pages={307--346},
  year={2015},
  publisher={APS}
}

@article{ji2024incompatibility,
  title={Incompatibility as a resource for programmable quantum instruments},
  author={Ji, Kaiyuan and Chitambar, Eric},
  journal={PRX Quantum},
  volume={5},
  number={1},
  pages={010340},
  year={2024},
  publisher={APS}
}

@article{zhang2017quantum,
  title={Quantum secure direct communication with quantum memory},
  author={Zhang, Wei and Ding, Dong-Sheng and Sheng, Yu-Bo and Zhou, Lan and Shi, Bao-Sen and Guo, Guang-Can},
  journal={Physical review letters},
  volume={118},
  number={22},
  pages={220501},
  year={2017},
  publisher={APS}
}

@article{baumgratz2014quantifying,
  title={Quantifying coherence},
  author={Baumgratz, Tillmann and Cramer, Marcus and Plenio, Martin B},
  journal={Physical review letters},
  volume={113},
  number={14},
  pages={140401},
  year={2014},
  publisher={APS}
}

@article{rosset2018resource,
  title={Resource theory of quantum memories and their faithful verification with minimal assumptions},
  author={Rosset, Denis and Buscemi, Francesco and Liang, Yeong-Cherng},
  journal={Physical Review X},
  volume={8},
  number={2},
  pages={021033},
  year={2018},
  publisher={APS}
}

@article{regula2021fundamental,
  title={Fundamental limitations on distillation of quantum channel resources},
  author={Regula, Bartosz and Takagi, Ryuji},
  journal={Nature Communications},
  volume={12},
  number={1},
  pages={4411},
  year={2021},
  publisher={Nature Publishing Group UK London}
}

@article{fang2022no,
  title={No-go theorems for quantum resource purification: New approach and channel theory},
  author={Fang, Kun and Liu, Zi-Wen},
  journal={PRX Quantum},
  volume={3},
  number={1},
  pages={010337},
  year={2022},
  publisher={APS}
}

@article{hedges2010efficient,
  title={Efficient quantum memory for light},
  author={Hedges, Morgan P and Longdell, Jevon J and Li, Yongmin and Sellars, Matthew J},
  journal={Nature},
  volume={465},
  number={7301},
  pages={1052--1056},
  year={2010},
  publisher={Nature Publishing Group UK London}
}

@article{choi2023unital,
  title={On unital qubit channels.},
  author={Choi, Man-Duen and Li, Chi-Kwong},
  journal={Quantum Inf. Comput.},
  volume={23},
  number={7\&8},
  pages={562--576},
  year={2023}
}

@book{wilde2013quantum,
  title={Quantum Information Theory},
  author={Wilde, Mark M},
  year={2013},
  publisher={Cambridge University Press}
}

@article{lewenstein1998separability,
  title={Separability and entanglement of composite quantum systems},
  author={Lewenstein, Maciej and Sanpera, Anna},
  journal={Physical Review Letters},
  volume={80},
  number={11},
  pages={2261},
  year={1998},
  publisher={APS}
}

@article{bennett1997capacities,
  title={Capacities of quantum erasure channels},
  author={Bennett, Charles H and DiVincenzo, David P and Smolin, John A},
  journal={Physical Review Letters},
  volume={78},
  number={16},
  pages={3217},
  year={1997},
  publisher={APS}
}

@article{streltsov2017colloquium,
  title={Colloquium: Quantum coherence as a resource},
  author={Streltsov, Alexander and Adesso, Gerardo and Plenio, Martin B},
  journal={Reviews of Modern Physics},
  volume={89},
  number={4},
  pages={041003},
  year={2017},
  publisher={APS}
}

@article{horodecki2009quantum,
  title={Quantum entanglement},
  author={Horodecki, Ryszard and Horodecki, Pawe{\l} and Horodecki, Micha{\l} and Horodecki, Karol},
  journal={Reviews of Modern Physics},
  volume={81},
  number={2},
  pages={865--942},
  year={2009},
  publisher={APS}
}

@article{Regula2017convex,
  title={Convex geometry of quantum resource quantification},
  author={Regula, Bartosz},
  journal={Journal of Physics A: Mathematical and Theoretical},
  volume={51},
  number={4},
  pages={045303},
  year={2017},
  publisher={IOP Publishing}
}

@article{bruss1998optimal,
  title={Optimal universal and state-dependent quantum cloning},
  author={Bru{\ss}, Dagmar and DiVincenzo, David P and Ekert, Artur and Fuchs, Christopher A and Macchiavello, Chiara and Smolin, John A},
  journal={Physical Review A},
  volume={57},
  number={4},
  pages={2368},
  year={1998},
  publisher={APS}
}

@article{takagi2019general,
  title={General resource theories in quantum mechanics and beyond: Operational characterization via discrimination tasks},
  author={Takagi, Ryuji and Regula, Bartosz},
  journal={Physical Review X},
  volume={9},
  number={3},
  pages={031053},
  year={2019},
  publisher={APS}
}

@article{chen2014comparison,
  title={Comparison of different definitions of the geometric measure of entanglement},
  author={Chen, Lin and Aulbach, Martin and Hajdu{\v{s}}ek, Michal},
  journal={Physical Review A},
  volume={89},
  number={4},
  pages={042305},
  year={2014},
  publisher={APS}
}

\fussy
\appendix
\section{Dual of the weight-based measure of channel}\label{sec:appdix_fa}
For a channel $\mathcal{N}$, the weight-based measure of quantum memory can be defined as
\begin{equation}\label{eq:a1}
C_w(\mathcal{N}) = \min\{s \ge 0 : \mathcal{N} \ge (1-s)\mathcal{M}, \mathcal{M} \in \operatorname{EB}\}.
\end{equation}
Equivalently, the weight-based measure can also be written via Choi matrix:
\begin{equation}\label{eq:a2}
C_w(\mathcal{N}) = \min\{s \ge 0 : \mathcal{J}_{\mathcal{N}} \ge (1-s)\mathcal{J}_{\mathcal{M}}, \mathcal{M} \in \operatorname{EB}\}.
\end{equation}
This can be recast as
\begin{equation}
\begin{aligned}
C_w(\mathcal{N})= & \min 1-\operatorname{Tr}[x_1] \\
\text { s.t. } & x_1+x_2=\mathcal{J}_N , \\
& x_1\in \operatorname{cone}(\mathfrak{F}),x_2\in\operatorname{cone}(\mathcal{V}) ,\\
\end{aligned}
\end{equation}
where we use $\mathfrak{F}$ to represent separable Choi states and $\mathcal{V}$ to represent bipartite Choi states, while $\operatorname{cone}(\mathfrak{F})$ and $\operatorname{cone}(\mathcal{V})$ represent their unnormalised versions.
Now define $W = \operatorname{cone}(\mathcal{V})\oplus\operatorname{cone}(\mathcal{V})$,$\mathcal{W'}=\mathcal{V}$,$\mathcal{K}=\operatorname{cone}(\mathfrak{F})\oplus\operatorname{cone}(\mathcal{V})$,$\Lambda(x_1\oplus x_2)=x_1+x_2$,$A=I\oplus0$,$y=\mathcal{J}_{\mathcal{N}}$. 

We can define the Lagrangian in the same formalRef.~\cite{takagi2019general} as
\begin{equation}\label{eq:a4}
\begin{aligned}
L(x;Q,Z)=1-\langle A,x\rangle-\langle Z,y-\Lambda(x)\rangle-\langle Q,x \rangle.
\end{aligned}
\end{equation}
where $Q\in\mathcal{W^*},Z\in\mathcal{W^{'*}}$ are the so-called Lagrange multipliers.
This allows us to write
\begin{equation}\label{eq:a5}
\begin{aligned}
\sup_{Q\in\mathcal{K^*},
Z\in W^{'*}} L(x ; Q, Z)=\left\{\begin{array}{ll}
1-\langle A, x\rangle & \text { if } \Lambda(x)=y \text { and } x \in \mathcal{K}, \\
\infty & \text { otherwise },
\end{array}\right.
\end{aligned}
\end{equation}
Noticing that
\begin{equation}\label{eq:a6}
\begin{aligned}
\inf _{x \in W} L(x ; Q, Z) & =\inf _{x \in W} 1-\langle Z, y\rangle +\langle \Lambda^{*}(Z)-A-Q, x\rangle ,\\
& =\left\{\begin{array}{ll}
1-\langle Z, y\rangle & \text { if } \Lambda^{*}(Z)-A-Q=0, \\
-\infty & \text { otherwise },
\end{array}\right.
\end{aligned}
\end{equation}
and the dual problem is
\begin{equation}\label{eq:a7}
\begin{aligned}
d & = \sup _{\substack{Q \in \mathfrak{F}^{*} \\ Z \in W^{\prime *}}} \inf _{x \in W} L(x ; Q, Z),\\
& =\sup \left\{1-\langle Z, y\rangle \mid \Lambda^{*}(Z)-A-Q=0, Q \in \mathcal{K}^{*}\right\},\\
& =\sup \left\{1-\langle Z, y\rangle \mid \Lambda^{*}(Z)-A \in \mathcal{K}^{*}\right\} .
\end{aligned}
\end{equation}
Then rewrite Eq.~\eqref{eq:a7} as:
\begin{equation}
\begin{aligned}
C_w(\mathcal{N})= & \max 1-\operatorname{Tr}[W\mathcal{J}_{\mathcal{N}}] ,\\
\text { s.t. } & \langle \Lambda^*(W)-I\oplus0,k\rangle\ge 0,\forall k\in \mathcal{K} , \\
& W\in \mathcal{W'} , \\
\end{aligned}
\end{equation}
where 
\begin{equation}
\begin{aligned}
\left\langle \Lambda^*(W)-I\oplus0,k\right\rangle &\ge 0 ,\forall k \in \mathcal{K} 
\Longleftrightarrow\\
&\operatorname{Tr}[Wx_1]+\operatorname{Tr}[Wx_2]\ge\operatorname{Tr}[x_1],\\
&\forall x_1\in\operatorname{cone}(\mathfrak{F}),\forall x_2 \in\operatorname{cone}(\mathcal{V}).
\end{aligned} 
\end{equation}
Note that the above condition is further equivalent to 
\begin{equation}
\begin{aligned}
\left(\operatorname{Tr}\left[W x_{1}\right] \ge \operatorname{Tr}\left[x_{1}\right], \forall x_{1} \in \operatorname{cone}(\mathfrak{F})\right) \wedge\\
\left(\operatorname{Tr}\left[W x_{2}\right] \geq 0, \forall x_{2} \in \operatorname{cone}(\mathcal{V})\right).
\end{aligned}
\end{equation}
We can set $x_2=0$ and $X_1=0$ to get the above results respectively. Therefore, the dual form can be written as
\begin{equation}\label{eq:a11}
\begin{aligned}
C_w(\mathcal{N})= & \max 1-\operatorname{Tr}\left[W \mathcal{J}_{\mathcal{N}}\right]\\
\text {s.t. } & \operatorname{Tr}\left[W x_{1}\right] \ge \operatorname{Tr}\left[x_{1}\right], \forall x_{1} \in \mathfrak{F} ,\\
& \operatorname{Tr}\left[W x_{2}\right] \geq 0, \forall x_{2} \in \mathcal{V} ,\\
& W \in W^{*},
\end{aligned}
\end{equation}
which can be expressed as
\begin{equation}
\begin{aligned}
C_w(\mathcal{N})= & \max 1-\operatorname{Tr}\left[W \mathcal{J}_{\mathcal{N}}\right] \\
\text {s.t. } & W^{\dagger}=W, \\
& \operatorname{Tr}[\mathcal{J}_{\mathcal{N}}W]\ge0,\\
& \operatorname{Tr}[\mathcal{J}_{\mathcal{M}}W]\ge1, \\
& \forall\mathcal{N}\in \operatorname{CPTP},\forall\mathcal{M}\in \operatorname{EB},
\end{aligned}
\end{equation}
where $\mathcal{J}_{\mathcal{N}}$ is bipartite Choi states and $\mathcal{J}_{\mathcal{M}}$ is separable Choi states.

\section{Maximum fidelity between maximally entangled states and separable states}\label{sec:appdix_fc}
In this section, we show that the relationship between the geometric measure of entanglement $E_G(\Psi^+)$ and the fidelity.

As Ref.~\cite{chen2014comparison} mentioned,
\begin{equation}
\begin{aligned}
\langle \Psi^+|\sigma|\Psi^+\rangle &= \sum_ip_i\langle \Psi^+|(\sigma_i^{(A)}\otimes\sigma_i^{(B)})|\Psi^+\rangle,\\
&\le \max_{i}\langle \Psi^+|(\sigma_i^{(A)}\otimes\sigma_i^{(B)})|\Psi^+\rangle.
\end{aligned}
\end{equation}
Any mixed separable state can be written as a convex combination of pure product states, and the maximum value is attained by pure product state $\sigma=\ket{a}\bra{a}\otimes\ket{b}\bra{b}$.

We can define the maximizing the fidelity between the input state and set of pure product states as following Ref.~\cite{chen2014comparison}:
\begin{equation}
\Lambda_{\mathrm{m}}^{2}(\rho):=\max _{\sigma \in \operatorname{SEP}} \operatorname{Tr}(\rho \sigma)=\max _{|\varphi\rangle \in \operatorname{PRO}}\langle\varphi| \rho|\varphi\rangle=\max _{|\varphi\rangle \in \mathrm{PRO}} F^{2}(\rho,|\varphi\rangle).
\end{equation}
Then we have 
\begin{equation}
\max_{\sigma\in SEP}\bra{\Psi^+}\sigma\ket{\Psi^+} = \max_{\substack{
  \ket{a} \in \mathcal{H}_A \\
  \ket{b} \in \mathcal{H}_B
}}|\langle{\Psi^+}|{a\otimes b}\rangle|^2.
\end{equation}
Consider normalized pure states 
\begin{equation}
    \ket{a}=\sum_i^{d}\alpha_i\ket{i}_A , \ket{b}=\sum_i^{d}\beta_i\ket{i}_B,
\end{equation}
where $\sum_i|\alpha_i|^2=\sum_i|\beta_i|^2=1$.
Then
\begin{equation}
\begin{aligned}
\langle{\Psi^+}|{a\otimes b}\rangle &= \langle \frac{1}{\sqrt{d}}\sum_i^d\ket{ii}|(\sum_{j=1}^d\alpha_j\ket{j}\otimes\sum_{k=1}^{d}\beta_k\ket{k})\rangle,\\
&=\frac{1}{\sqrt{d}}\sum_i^{d}\alpha_i\beta_i.
\end{aligned}
\end{equation}
Thus
\begin{equation}
\begin{aligned}
|\langle{\Psi^+}|{a\otimes b}\rangle|^2 &= \frac{1}{d}  |\sum_i^{d}\alpha_i\beta_i|^2 ,\\
&\le \frac{1}{d}(\sum_{i=1}^d|\alpha_i|^2)(\sum_{i=1}^d|\beta_i|^2),\\
&=\frac{1}{d}.
\end{aligned}
\end{equation}
The inequality holds is due to Cauchy-Schwarz inequality.

\section{The proof of $\operatorname{Tr}[AB]\le\lambda_{\max}(A)\operatorname{Tr}[B]$}\label{sec:appdix_fd}
Suppose $A$ is a Hermitian matrix and $B$ is a positive semi-definite matrix. Since $B$ is a positive semi-definite operator, it admits a spectral decomposition
\begin{equation}
     B = \sum_{i} \mu_i |u_i\rangle\langle u_i|,
\end{equation}
where $\mu_i \ge 0$ are the eigenvalues of $B$, and $\{|u_i\rangle\}$ is a corresponding orthonormal basis of eigenvectors. Thus, we can get the trace of $B$ is the sum of its eigenvalues $\operatorname{Tr}[B] = \sum_{i} \mu_i$.
Now we compute :
\begin{equation}
\begin{aligned}
    \operatorname{Tr}[AB] &= \operatorname{Tr}\left[A \left(\sum_{i} \mu_i |u_i\rangle\langle u_i|\right)\right],\\
    &= \sum_{i} \mu_i \operatorname{Tr}[A|u_i\rangle\langle u_i|],\\
    &= \sum_{i} \mu_i \langle u_i|A|u_i\rangle.
\end{aligned}
\end{equation}
Here, $\langle u_i|A|u_i\rangle$ is the expectation value of the operator $A$ in the state $|u_i\rangle$.
For any Hermitian operator $A$ and any normalized vector $|u\rangle$, its expectation value $\langle u|A|u\rangle$ is upper-bounded by the largest eigenvalue of $A$, denoted
\begin{equation}
    \langle u_i|A|u_i\rangle \le \lambda_{\max}(A).
\end{equation}
Thus, we can get 
\begin{equation}
\begin{aligned}
    \operatorname{Tr}[AB]&=\sum_{i} \mu_i \langle u_i|A|u_i\rangle,\\
    & \le \sum_{i} \mu_i \lambda_{\max}(A),\\
    &=\lambda_{\max}(A)\operatorname{Tr}[B].
\end{aligned}
\end{equation}
This completes the proof.
\end{document}